\documentclass[11pt,letterpaper]{article}

\usepackage[utf8]{inputenc}
\usepackage[english]{babel}
\usepackage{csquotes}
\usepackage{CJKutf8}

\usepackage[margin=1in]{geometry} 

\usepackage{setspace}
\usepackage{graphicx}

\usepackage{amsmath, amssymb, amsfonts}
\usepackage{amsthm, mathrsfs, mathtools}
\usepackage{algorithm}
\usepackage{algorithmic}

\usepackage{hyperref}
\usepackage[capitalise]{cleveref}

\usepackage[backend=biber,style=alphabetic, minbibnames=99,sorting=nyt,sortcites=true,maxbibnames=99]{biblatex}
\usepackage{titling}

\usepackage{authblk}

\posttitle{\par\end{center}\vspace{-0.5em}}

\usepackage{enumitem} 
\usepackage{caption} 
\usepackage{subcaption} 
\usepackage{booktabs} 
\usepackage{microtype} 
\usepackage{dsfont} 
\usepackage{faktor} 
\usepackage{tikz}
\usetikzlibrary{shapes.geometric} 
\usepackage[baseline]{euflag}

\let\epsilon\varepsilon

\let\phi\varphi

\theoremstyle{plain}
\newtheorem{theorem}{Theorem}[section]
\newtheorem{lemma}[theorem]{Lemma}
\newtheorem{corollary}[theorem]{Corollary}
\newtheorem{proposition}[theorem]{Proposition}

\newtheorem{definition}[theorem]{Definition}
\newtheorem{example}[theorem]{Example}
\newtheorem{remark}[theorem]{Remark}
\newtheorem*{theorem*}{Theorem}
\newtheorem{maintheorem}{Theorem}

\newcommand*{\Tr}[1]{\mathop{}\!\mathrm{Tr}{\left(#1\right)}}

\def\id{\mathds{1}}
\DeclareMathOperator{\Span}{span}
\DeclareMathOperator{\Enc}{Enc}
\DeclareMathOperator{\Dec}{Dec}
\DeclareMathOperator{\Gen}{Gen}
\DeclareMathOperator{\Eval}{Eval}

\newcommand{\poly}{\mathrm{poly}}
\newcommand{\op}{\mathrm{op}}
\newcommand{\bra}[1]{\ensuremath{\left\langle#1\right|}}
\newcommand{\ket}[1]{\ensuremath{\left|#1\right\rangle}}

\newcommand{\sandwich}[3]{\langle #1|#2 |#3\rangle  }
\newcommand{\norm}[1]{\left\lVert#1\right\rVert}
\newcommand{\abs}[1]{\left\lvert#1\right\rvert}

\newcommand{\cH}{\mathcal{H}}

\newcommand{\cA}{\mathcal{A}}

\newcommand{\cE}{\mathcal{E}}

\newcommand{\cK}{\mathcal{K}}
\newcommand{\cR}{\mathcal{R}}
\newcommand{\cG}{\mathcal{G}}
\newcommand{\cW}{\mathcal{W}}

\let\B\relax 
 
\newcommand{\A}{A_{a|x}}
\newcommand{\B}{B_{b|y}}

\usepackage{pict2e,picture}
\makeatletter
\newcommand{\avsum}{%
  \DOTSB\mathop{\mathpalette\@avsum\relax}\slimits@
}
\newcommand{\@avsum}[2]{%
  \begingroup
  \sbox\z@{$#1\sum$}%
  \setlength{\unitlength}{%
    \dimexpr
      \ifx#1\displaystyle1\else3\fi\dimexpr\@avsumthickness{#1}\relax+
      \ht\z@+\dp\z@
    \relax
  }%
  \linethickness{\@avsumthickness{#1}}%
  \vphantom{\sum}%
  \smash{\ooalign{$\m@th#1\sum$\cr\hidewidth$#1\@avsumbar$\hidewidth\cr}}%
  \endgroup
}
\newcommand{\@avsumbar}{%
  \vcenter{\hbox{\begin{picture}(0,1)\roundcap\Line(0,0)(0,1)\end{picture}}}%
}
\newcommand{\@avsumthickness}[1]{
  1.25\fontdimen8
    \ifx#1\displaystyle\textfont\else
    \ifx#1\textstyle\textfont\else
    \ifx#1\scriptstyle\scriptfont\else
    \scriptscriptfont\fi\fi\fi 3
}
\makeatother

\newcommand{\ax}{{a|x}}
\newcommand{\by}{{b|y}}
\newcommand{\cz}{{c|z}}
\newcommand{\abxy}{{ab|xy}}
\newcommand{\bcyz}{{bc|yz}}

\newcommand{\abcxyz}{{abc|xyz}}
\newcommand{\akxk}{{a_{[k]} | x_{[k]}}}
\newcommand{\QInstrument}{\mathsf{T}}

\DeclareMathOperator{\negl}{negl}

\newcommand{\NPA}{\mathrm{NPA}}
\newcommand{\strict}{\mathrm{strict}}

\newcommand{\comp}{\mathrm{comp}}
\newcommand{\gamevalueqcopt}[1]{\omega_{\mathrm{qc}}(#1)} 
\newcommand{\gamevalueNPA}[2]{\omega_{\mathrm{NPA}}^{#2}(#1)} 
\newcommand{\gamevalueSeqNPA}[2]{\omega_{\mathrm{seqNPA}}^{#2}(#1)} 

\newcommand{\gamevalueBiSeqNPA}[2]{\omega_{2\mathrm{seqNPA}}^{#2}(#1)} 
\newcommand{\gamevalueTriSeqNPA}[2]{\omega_{3\mathrm{seqNPA}}^{#2}(#1)} 
\newcommand{\gamevaluekSeqNPA}[2]{\omega_{k\mathrm{seqNPA}}^{#2}(#1)} 
\newcommand{\gamevalueCompile}[2]{\omega_{#1}({#2})} 

\usepackage{xcolor}

\hypersetup{
  pdftitle    = {Quantitative quantum soundness for all multipartite compiled nonlocal games},
  pdfauthor   = {Matilde Baroni; Igor Klep; Dominik Leichtle; Marc-Olivier Renou; Ivan Šupić; Lucas Tendick; Xiangling Xu},
  pdfsubject  = {Quantum information theory; cryptography},
  pdfkeywords = {compiled nonlocal games, multipartite nonlocal game, quantitative quantum soundness, sequential NPA hierarchy, NPA hierarchy, quantum homomorphic encryption, weakly non-signaling decomposition},
  pdflang     = {en},
  pdfcreator  = {LaTeX with hyperref},
  pdfdisplaydoctitle = true
}

\title{\bfseries{Quantitative quantum soundness for all multipartite compiled nonlocal games}}

\author[1$\dagger$]{Matilde Baroni}
\author[2,3,4]{Igor Klep}
\author[5]{Dominik Leichtle}
\author[6,7,8]{Marc-Olivier Renou}
\author[9]{Ivan \v{S}upi\'{c}}
\author[6,7,8]{Lucas Tendick}
\author[6,7,8*]{Xiangling Xu}

\affil[1]{Sorbonne Universit\'e, CNRS, LIP6, 4 place Jussieu, 75005 Paris, France}
\affil[2]{Faculty of Mathematics and Physics, University of Ljubljana, Slovenia}
\affil[3]{Faculty of Mathematics, Natural Sciences
and Information Technologies, 
University of Primorska, Slovenia}
\affil[4]{Institute of Mathematics, Physics and Mechanics, Ljubljana, Slovenia}
\affil[5]{School of Informatics, University of Edinburgh, 10 Crichton Street, Edinburgh EH8 9AB, United Kingdom}
\affil[6]{Inria Paris-Saclay, B\^atiment Alan Turing, 1 rue Honor\'e d'Estienne d'Orves, 91120 Palaiseau, France}
\affil[7]{CPHT, Ecole polytechnique, Institut Polytechnique de Paris, Route de Saclay, 91128 Palaiseau, France}
\affil[8]{LIX, Ecole polytechnique, Institut Polytechnique de Paris, Route de Saclay, 91128 Palaiseau, France}
\affil[9]{Universit\'e Grenoble Aples, CNRS, Grenoble INP, LIG, 38000 Grenoble, France}
\affil[$\dagger$]{\scriptsize\texttt{matilde.baroni@lip6.fr}}
\affil[*]{\scriptsize\texttt{xu.xiangling@inria.fr}}

\date{\vspace{-3.90em}}

\allowdisplaybreaks
\begin{document}
\maketitle
\begin{abstract}
Compiled nonlocal games transfer the power of Bell-type multi-prover tests into a single-device setting by replacing spatial separation with cryptography.
Concretely, the KLVY compiler (STOC'23) maps any multi-prover game to an interactive single-prover protocol, using quantum homomorphic encryption.
A crucial security property of such compilers is quantum soundness, which ensures that a dishonest quantum prover cannot exceed the original game's quantum value.
For practical cryptographic implementations, this soundness must be quantitative, providing concrete bounds rather than merely asymptotic.
While quantitative quantum soundness has been established for the KLVY compiler in the bipartite case, it has only been shown asymptotically for multipartite games.
This is a significant gap, as multipartite nonlocality exhibits phenomena with no bipartite analogue, and the difficulty of enforcing space-like separation makes single-device compilation especially compelling.
This work closes this gap by providing quantitative upper bounds for \emph{all multipartite} compiled nonlocal games via a new sequential NPA-like hierarchy.
In particular, finite-level convergence yields \emph{quantitative quantum soundness} with respect to the commuting quantum value, and flat optimality yields the same with respect to the tensor-product quantum value.
On the way, we introduce an \emph{NPA-like hierarchy for quantum instruments} and prove its completeness, thereby characterizing correlations from operationally-non-signaling sequential strategies.
This NPA-like hierarchy can be seen to complement previous multipartite generalizations of the S-G-HJW purification theorem, which takes a central role in quantum information, nonlocality, and contextuality.
We further develop novel geometric arguments for the decomposition of sequential strategies into their signaling and non-signaling parts, which might be of independent interest.
\end{abstract}


\section{Introduction}\label{sec:Introduction}
Nonlocal games are cooperative tasks involving multiple, non-communicating players (provers) and a referee (verifier), see \cref{fig:NonlocalCompiledBellGame}(a), originally designed to test the foundational limits of classical physics in Bell's seminal work~\cite{bell1964einstein}.
In this setting, provers who share quantum entanglement can coordinate their answers to the verifier's questions in ways that are provably impossible using only classical resources.
This ``quantum advantage'' makes nonlocal games powerful tools for the device-independent (i.e., black-box)  certification of quantum properties, such as entanglement, which find application, for instance, in the generation of genuine randomness ~\cite{scarani2012device,brunner2014bell}.
The security of these protocols, however, relies on a crucial assumption: that the provers cannot communicate.
Typically, this is enforced by physical (space-like) separation, an approach that is experimentally demanding~\cite{hensen2015loophole,giustina2015significant} and scales poorly, especially as the number of provers increases \cite{Horodecki2019}.
This raises a fundamental question in both theory and practice: can we leverage the power of nonlocality within a protocol involving only a \emph{single}, untrusted quantum device?

The naive approach of having one prover who plays all the roles fails, as the prover would see all questions and trivially bypass the no-communication constraint.
Cryptography provides an alternative solution, where the no-communication condition is replaced by computational assumptions on the single prover~\cite{kalai2023quantum,huang2025compilingmathsfmipsuccinctclassical,merkulov2025computationalbellinequalities}.
The focus of our work is the KLVY compiler~\cite{kalai2023quantum}, which introduces a procedure to transform any $k$-player nonlocal game into a sequential, multi-round protocol between a verifier and a single prover, using quantum homomorphic encryption (QHE)~\cite{mahadev2020classical,brakerski2018quantum}; see \cref{fig:NonlocalCompiledBellGame}(c).
Intuitively, the QHE allows the prover to compute on encrypted questions, preventing them from learning the questions for early players while generating a valid answer for a later player, up to negligible probability under standard cryptographic assumptions.
The KLVY compiler~\cite{kalai2023quantum} is known to be (i) \emph{classically sound} (a classical prover cannot perform better than in the original game) and (ii) \emph{quantum complete} (an honest quantum prover can achieve the optimal quantum score).
The central, unresolved security question is that of \emph{quantum soundness}: can a dishonest quantum prover exploit the compiled protocol to exceed the score achievable in the original, spatially-separated game?

Progress on this question has been built up through a series of works of increasing generality.
Initial results established soundness for specific bipartite (two-player) games~\cite{natarajan2023bounding, cui2024computational,baroni2024quantum,mehta2024self}.
Subsequently,~\cite{kulpe2024bound} proved asymptotic soundness for all bipartite games, showing that the compiled game's score approaches the true quantum score in the limit of $\lambda \to \infty$ for the cryptographic security parameter.
While theoretically significant, this guaranty is insufficient for practical cryptography, which requires \emph{quantitative} bounds for a finite, fixed security parameter $\lambda$.

Furthermore, this bipartite framework does not extend to the multipartite setting, leaving a crucial gap.
At the foundational level, the quantum soundness of the compilation of nonlocal games enables to confirm our understanding of the exact nature of the constraints imposed by space-like separation in quantum information theory. 
The fact that the scores of the compiled and non-compiled versions of the game are the same demonstrates that physical space-like separation can be simulated cryptographically only. Conversely, if some game could not be compiled in a sound way, it would suggest that space-like separation imposes additional constraints that are not well understood.  

At a more practical level, in quantum information, genuinely multipartite effects are necessary to account for new phenomena \cite{no-bipartite-nonlocal}, e.g., post-quantum steering \cite{sainz2015postquantum}, network non-locality \cite{genuine-network-multipartite}, and multipartite self-testing \cite{st-multipartite-network} has been proved to be useful. The multipartite setting is also central in computer science. Many verification tasks naturally involve more than two parties: the verifier must often check that several local views are consistent with one global object, a scenario well captured by multiplayer games. At the intersection of these fields, multipartite games reveal qualitatively new behavior: for instance, the quantum-classical gap amplification behaves differently for more than two players, for some relevant families of games parallel repetition can quickly shrink the cheating probability in the multipartite setting, but the general parallel repetition theorem for the quantum value of multipartite nonlocal games is still an open problem.
Overall, working in the multipartite model keeps alphabets small, queries local, and completeness–soundness gaps large, yielding cleaner designs, stronger guarantees, and more flexible composition tools: the very ingredients that drive verification, hardness, and protocol design~\cite{natarajan2017quantum,mousavi2022nonlocal,ito2012multi,grilo2019perfect}.

\begin{figure}[ht]
  \centering
  \begin{subfigure}[t]{0.30\textwidth}
    \centering
    \scalebox{0.6}{\begin{tikzpicture}[scale=0.3]
    \fill[gray!30] (0,0) -- (24,0) -- (24,1) -- (0,1) -- cycle;
    \fill[gray!30] (0,0) -- (1,0) -- (1,10) -- (0,10) -- cycle;
    \fill[gray!30] (0,9) -- (24,9) -- (24,10) -- (0,10) -- cycle;

    \fill[orange!30] (4,3) rectangle (8,7);
    \fill[orange!30] (12,3) rectangle (16,7);
    \fill[orange!30] (20,3) rectangle (24,7);

    \draw[->, thick] (6,9) -- (6,7) node[right,midway] {$x$};
    \draw[->, thick] (14,9) -- (14,7)node[right,midway] {$y$};
    \draw[->, thick] (22,9) -- (22,7) node[right,midway] {$z$};
    
    \draw[->, thick] (6,3) -- (6,1)node[right,midway] {$a$};
    \draw[->, thick] (14,3) -- (14,1)node[right,midway] {$b$};
    \draw[->, thick] (22,3) -- (22,1)node[right,midway] {$c$};

    \fill[orange!30] (2,11) -- (20,11) -- (20,13) -- (2,13) -- cycle;
    \fill[orange!30] (2,3) -- (4,3) -- (4,11) -- (2,11) -- cycle;
    \fill[orange!30] (10,3) -- (12,3) -- (12,11) -- (10,11) -- cycle;
    \fill[orange!30] (18,3) -- (20,3) -- (20,11) -- (18,11) -- cycle;

    \node at (5,5){Alice};
    \node at (13,5) {Bob};
    \node at (21,5) {Charlie};

\end{tikzpicture}}
    \subcaption{Three spatially separated provers, Alice $A$, Bob $B$, and Charlie $C$, receive questions $x,y,z$ and return answers $a,b,c$. The players' shared strategy is defined by the correlations $p(abc|xyz)$ with the corresponding score (winning probability) given by $\sum_{a,b,c,x,y,z} \beta_{abcxyz} p(abc|xyz)$, where $\beta_{abcxyz}$ is the payoff tensor associated with the rule of the game $\cG$. By using a shared entangled state, the players can generate quantum correlations leading to scores that provably exceed classical limitations. The maximum achievable quantum score is denoted by $\omega_{\mathrm{q}}(\cG)$ (the tensor-product quantum score).}
    \label{fig:nonlocal-a}
  \end{subfigure}
  \hfill
  \begin{subfigure}[t]{0.23\textwidth}
    \centering
    \scalebox{0.6}{\begin{tikzpicture}[scale=0.3]
    \fill[gray!30] (8.5,0) -- (9.5,0) --(9.5,28) -- (8.5,28) -- cycle;
    \fill[gray!30] (3,0) -- (8.5,0) --(8.5,1) -- (3,1) -- cycle;
    \fill[gray!30] (3,9) -- (8.5,9) --(8.5,10) -- (3,10) -- cycle;
    \fill[gray!30] (3,18) -- (8.5,18) --(8.5,19) -- (3,19) -- cycle;
    \fill[gray!30] (3,27) -- (8.5,27) --(8.5,28) -- (3,28) -- cycle;

    \fill[orange!30] (0,3) rectangle (6,7);
    \fill[orange!30] (0,12) rectangle (6,16);
    \fill[orange!30] (0,21) rectangle (6,25);

    \fill[orange!60] (0.5,12) rectangle (1.5,8);
    \node[isosceles triangle,
	isosceles triangle apex angle=90,
	fill=orange!60,
        rotate=-90,
	minimum size =5] (T90)at (1,7.8){};

        \fill[orange!60] (0.5,21) rectangle (1.5,17);
    \node[isosceles triangle,
	isosceles triangle apex angle=90,
	fill=orange!60,
        rotate=-90,
	minimum size =5] (T90)at (1,16.8){};


    \node at (3,5){Charlie};
    \node at (3,14) {Bob};
    \node at (3,23) {Alice};

    \draw[->, thick] (4,27) -- (4,25) node[right,midway] {$x$};
    \draw[->, thick] (4,18) -- (4,16)node[right,midway] {$y$};
    \draw[->, thick] (4,9) -- (4,7) node[right,midway] {$z$};
    
    \draw[->, thick] (4,21) -- (4,19)node[right,midway] {$a$};
    \draw[->, thick] (4,12) -- (4,10)node[right,midway] {$b$};
    \draw[->, thick] (4,3) -- (4,1)node[right,midway] {$c$};


	isosceles triangle apex angle=90,
	fill=orange!30,
        rotate=-90,
	minimum size =5] (T90)at (6,0){};
    
\end{tikzpicture}}
    \subcaption{The game is played in sequence: Alice acts first with $(x,a)$, then Bob with $(y,b)$, then Charlie with $(z,c)$. In the Heisenberg algebraic picture, a state $\sigma$ is fixed. Actions of $A, B$ are described by quantum instruments (completely positive maps) $\{\QInstrument_{a|x}\}_a$ and $\{\QInstrument_{b|y}\}_b$, while $C$ measures with POVM effects $\{f_{c|z}\}_c$ as usual. The resulting correlations are $p(abc|xyz)=\Tr{\sigma \cdot \QInstrument_{a|x} \circ \QInstrument_{b|y}(f_{c|z})}$. \emph{Operational-non-signaling} requires $\sum_a \QInstrument_{a|x}=\sum_a \QInstrument_{a|x'}$ for all $x,x'$ and $\sum_b \QInstrument_{b|y}=\sum_b \QInstrument_{b|y'}$ for all $y,y'$.}
    \label{fig:nonlocal-b}
  \end{subfigure}
  \hfill
  \begin{subfigure}[t]{0.23\textwidth}
    \centering
    \scalebox{0.6}{\begin{tikzpicture}[scale=0.3]
    \fill[gray!30] (8.4,0) -- (9.4,0) --(9.4,28) -- (8.4,28) -- cycle;
    \fill[gray!30] (3,0) -- (8.5,0) --(8.5,1) -- (3,1) -- cycle;
    \fill[gray!30] (3,9) -- (8.5,9) --(8.5,10) -- (3,10) -- cycle;
    \fill[gray!30] (3,18) -- (8.5,18) --(8.5,19) -- (3,19) -- cycle;
    \fill[gray!30] (3,27) -- (8.5,27) --(8.5,28) -- (3,28) -- cycle;

    \fill[orange!30] (2,3) rectangle (6,7);
    \fill[orange!30] (2,12) rectangle (6,16);
    \fill[orange!30] (2,21) rectangle (6,25);

    \fill[orange!30] (0,3) -- (2,3) -- (2,25) -- (0,25) -- cycle;

    \node at (3,5){Charlie};
    \node at (3,14) {Bob};
    \node at (3,23) {Alice};

    \draw[->, thick] (4,27) -- (4,25) node[right,midway] {$\mathsf{Enc}(x)$};
    \draw[->, thick] (4,18) -- (4,16)node[right,midway] {$\mathsf{Enc}(y)$};
    \draw[->, thick] (4,9) -- (4,7) node[right,midway] {$z$};
    
    \draw[->, thick] (4,21) -- (4,19)node[right,midway] {$\mathsf{Enc}(a)$};
    \draw[->, thick] (4,12) -- (4,10)node[right,midway] {$\mathsf{Enc}(b)$};
    \draw[->, thick] (4,3) -- (4,1)node[right,midway] {$c$};

    
\end{tikzpicture}}
    \subcaption{A single prover $P$ plays all roles sequentially. The pairs $(x,a)$ and $(y,b)$ are sent and returned in encrypted form $\mathrm{Enc}(x),\mathrm{Enc}(a)$ and $\mathrm{Enc}(y),\mathrm{Enc}(b)$ (the prover computes on the encrypted data homomorphically), while $(z,c)$ is sent in the clear. Security of the QHE scheme enforces computational non-signaling between the $A$-, $B$-, and $C$-interfaces.}
    \label{fig:nonlocal-c}
  \end{subfigure}
  \caption{Nonlocal game variants: (a) standard, (b) sequential, (c) compiled. Time flows from top to bottom.}
  \label{fig:NonlocalCompiledBellGame}
\end{figure}

Two recent and independent works began to close these gaps.
The authors of \cite{klep2025quantitative} established the first quantitative soundness bounds for bipartite games via finite-level certificates in a sequential variant of the Navascu\'es-Pironio-Ac\'in (NPA) hierarchy~\cite{navascues2008convergent,pironio2010convergent}; in particular, flat optimality gives bounds with respect to the tensor-product quantum value.
See also~\cite{cui2025convergent} for a parallel independent work that approaches the same problem from the dual sums-of-squares (SOS) perspective.
In parallel,~\cite{baroni2025asymptotic} proved asymptotic quantum soundness for all multipartite games, developing composable tools that fully generalize the algebraic approach of \cite{kulpe2024bound}, setting up a strong mathematical foundation for future work.

However, a unified solution for the general multipartite case remains challenging, as both previous results use different (a priori not compatible) techniques, leaving the question of whether quantitative quantum soundness for all multipartite games can be achieved open.
More precisely, the sequential NPA hierarchy of~\cite{klep2025quantitative} and the nice SOS hierarchy of~\cite{cui2025convergent} are both tailored to two-player sequential games and does not naturally accommodate the complex algebraic structure of multipartite interactions of~\cite{baroni2025asymptotic}.
Moreover, their core analytical tools are intrinsically bipartite and does not generalize.
This is a significant obstacle, as multipartite nonlocality exhibits richer quantum phenomena than its bipartite counterpart, and the practical difficulty of enforcing physical separation among many parties makes a robust cryptographic compilation especially desirable.

This work overcomes the above obstacles.
That is, we answer the question of whether quantitative upper bounds can be obtained for general multipartite compiled nonlocal games, and obtain quantum-value soundness whenever the associated sequential hierarchy has a finite-level certificate.
This makes the design of new multi-prover interactive proofs more compelling, knowing that such games can be compiled and experimentally realized in a single prover setting.
Our main tool is a novel and composable NPA-like hierarchy designed to model the sequential application of quantum instruments.
This framework, combined with new geometric proof techniques for the weak signaling decomposition, is rich enough to capture general multipartite scenarios in the compiled setting and may be of independent interest for the study of complex quantum protocols.

\subsection{Main results}\label{sec:IntroMainResult}
Our first main result is a quantitative quantum soundness theorem for multipartite compiled nonlocal games, providing bounds for finite security levels.
For a given security parameter $\lambda$, we define an \emph{efficient} prover as one implementable in quantum-polynomial time (QPT), i.e., by a quantum circuit of size $\poly(\lambda)$.
We call a function \emph{negligible}, denoted $\negl(\lambda)$, if it vanishes faster than the reciprocal of any polynomial in $\lambda$.

\begin{maintheorem}[\cref{cor:kpartite_QuantumSoundness}]\label{thm:MainQuantumSoundness}
    Let $k, \lambda \in \mathds{N}$, let $\cG$ be a $k$-partite nonlocal game with optimal commuting-operator quantum score $\gamevalueqcopt{\cG}$, and let $\cG_{\comp}$ be its compiled version.
    Let $S = (S_{\lambda})_{\lambda}$ be any strategy employed by an efficient prover, and denote by $\gamevalueCompile{\lambda}{\cG_{\comp}, S}$ its compiled Bell score.

    If the $k$-partite sequential NPA hierarchy (see \cref{eq:kpartite_SequentialNPASDP} for the full definition) for $\cG$ admits a flat optimal solution at some level $n$, then there exists a negligible function $\negl_{S}(\lambda)$ (depending on the QHE scheme and on $S$) such that
    \begin{align}\label{eq:MainthmFiniteDimGameUpperBound}
        \gamevalueCompile{\lambda}{\cG_{\comp}, S} \leq \omega_{\mathrm{q}}(\cG) + \negl_{S}(\lambda),
    \end{align}
    where $\omega_{\mathrm{q}}(\cG)$ is the optimal tensor-product quantum score.

    More generally, for every $k$-partite nonlocal game $\cG$ and for every $n \in \mathds{N}$, there exists a negligible function $\negl_{S, n}(\lambda)$ (depending on the QHE scheme, the strategy $S$, and $n$) such that
    \begin{align}\label{eq:MainthmAllGameUpperBound}
        \gamevalueCompile{\lambda}{\cG_{\comp}, S} \leq \gamevaluekSeqNPA{\cG}{n} + \negl_{S, n}(\lambda),
    \end{align}
    where $\gamevaluekSeqNPA{\cG}{n}$ is the level-$n$ value of the $k$-partite sequential NPA hierarchy, such that $\gamevaluekSeqNPA{\cG}{n} \searrow \omega_{\mathrm{qc}}({\cG})$ as $n \to \infty$.
    In particular, if the hierarchy converges at some finite $n$, i.e., $\gamevaluekSeqNPA{\cG}{n} = \omega_{\mathrm{qc}}({\cG})$, then we recover \cref{eq:MainthmFiniteDimGameUpperBound} with $\omega_{\mathrm{q}}(\cG)$ replaced by $\omega_{\mathrm{qc}}({\cG})$.
\end{maintheorem}

A key observation is that the error terms $\negl_{S,n}(\lambda)$ (and, when applicable, $\negl_S(\lambda)$) are determined solely by the compiled strategy $S$ and the security of the underlying QHE scheme.
In particular, they do not rely on any convergence guarantee for the sequential hierarchy, nor on our ability to solve it (see \cref{rem:NegligibleFunctionIndependent}).

This theorem recovers the quantitative bipartite result of~\cite{klep2025quantitative} for $k=2$.
In the limit where $\lambda\to\infty$ and $n\to\infty$, it reproduces the asymptotic multipartite result of~\cite{baroni2025asymptotic}.
Thus, \cref{thm:MainQuantumSoundness} unifies all prior work, delivering quantitative quantum soundness under finite-level hierarchy certificates, a commuting-quantum-value upper bound in the asymptotic case, and systematic hierarchy-based bounds in full generality.

The central technical tool enabling this result is a novel sequential NPA hierarchy, for which we show completeness, strict feasibility, and a stopping criterion for finite convergence.
\begin{maintheorem}[Informal, \cref{thm:kpartite_Completeness,prop:kpartite_FlatnessCondition,prop:kpartite_StrictFeasible}]\label{thm:MainCompleteness}
    For any $k$, the $k$-partite sequential NPA hierarchy (\cref{eq:kpartite_SequentialNPASDP}) is strictly feasible, and is complete with respect to $k$-partite commuting-observable strategies (and thus $k$-partite operationally-non-signaling sequential strategies), i.e., the sequence of its finite-level values satisfies
    \begin{align*}
        \gamevaluekSeqNPA{\cG}{n} \searrow \omega_{\mathrm{qc}}({\cG}) \quad \text{ as } n\to\infty.
    \end{align*}
    Moreover, if $\cG$ admits a flat optimal solution to \cref{eq:kpartite_SequentialNPASDP} at some finite level $n$, then this flat solution gives rise to a finite-dimensional optimal tensor-product quantum strategy, and the hierarchy converges at level $n$ to the tensor product quantum value, i.e., $\gamevaluekSeqNPA{\cG}{n} = \omega_{\mathrm{qc}}(\cG) = \omega_{\mathrm q}(\cG)$
\end{maintheorem}

\Cref{thm:MainCompleteness} can be interpreted as providing a concrete NPA hierarchy for the multipartite generalization of the bipartite Schrödinger-Gisin-Hughston-Jozsa-Wootters (S-G-HJW) purification theorem \cite{Schrödinger_1936,gisin1989stochastic,hughston1993complete}, which was recently proved by~\cite{baroni2025asymptotic}.
This generalization asserts that the correlations in operationally-non-signaling sequential scenarios coincide with those realized by commuting-observable strategies, for any number of players.
As such, our novel NPA hierarchy generalizes the NPA hierarchy that was previously considered for the bipartite sequential case in~\cite{klep2025quantitative,cui2025convergent}

\subsection{Methods and techniques}\label{sec:IntroMethodsTechnique}
\noindent
\textbf{Previous works on bipartite games.}
As in \cref{thm:MainQuantumSoundness}, our primary goal is to characterize and compute upper bounds on the scores of compiled nonlocal games.
Recent works \cite{kulpe2024bound, baroni2025asymptotic} established a crucial connection: in the asymptotic limit of perfect cryptographic security, compiled strategies (\cref{fig:nonlocal-c}) are equivalent to sequential strategies (\cref{fig:nonlocal-b}), where players are constrained to be operationally--non-signaling.
That is, the quantum operations of the players, when averaged over all classical outcomes, reveal no information about their classical input.
This provides a path to bounding compiled scores by analyzing their sequential counterparts.

However, for any real-world implementation with a finite security parameter, this correspondence is not exact, and a quantitative analysis is required.
The framework of \cite{klep2025quantitative} provided the first quantitative bounds for bipartite nonlocal games as follows: it introduces a sequential NPA hierarchy that is complete for the bipartite quantum scenarios.
Then, they show that any efficient strategy for a compiled game is necessarily ``close'' to a feasible solution of this hierarchy, using a signaling decomposition argument.
By proving this, the score of the compiled strategy becomes bounded by the value of the sequential hierarchy, thus demonstrating soundness due to its convergence.

\vspace{0.3em}
\noindent
\textbf{Limitations of the previous bipartite techniques.}
This bipartite solution, however, faces two fundamental obstacles that prevent its generalization to multipartite cases:
\begin{enumerate}[noitemsep, topsep=2pt, leftmargin=*]
    \item \emph{An unextendable hierarchy:} The bipartite hierarchies in~\cite{klep2025quantitative,cui2025convergent} are built on the Schr\"odinger picture, modeling the post-measurement state passed between players.
    While sufficient for two players, this approach loses critical algebraic information and is known to fail in multipartite scenarios~\cite{sainz2015postquantum}.
    Multipartite sequential protocols require a more robust model that can handle a chain of quantum instruments (CP maps).
    Specifically, the moment matrices defined in~\cite[Eq.~22]{klep2025quantitative} and~\cite[Eqs.~4.2 and~4.4]{cui2025convergent} exhibit a \emph{block-diagonal} structure indexed by Alice's inputs and outputs $(a,x)$.
    While the information contained in these diagonal blocks suffices for bipartite scenarios, this structure implicitly discards \emph{off-diagonal terms} that are critical for multipartite cases~\cite{sainz2015postquantum,baroni2025asymptotic}.
    Consequently, multipartite sequential protocols require a more robust model capable of handling a chain of quantum instruments (CP maps) to preserve this essential algebraic information.
    \item \emph{Non-scalable proof techniques:} The proof arguments in previous works are intrinsically bipartite.
    The signaling decomposition in~\cite{klep2025quantitative} relies on a Gelfand-Naimark-Segal (GNS) representation of the underlying algebra, while the key lemma~\cite[Lemma~3.4]{cui2025convergent} of the SOS approach relies on a calculation that is similarly limited to bipartite scenarios.
    Neither argument generalizes to the multipartite setting, necessitating a completely new approach.
\end{enumerate}

\begin{figure}
    \centering
    \resizebox{0.7\linewidth}{!}{\begin{tikzpicture}[scale=0.8,every node/.append style={scale=0.8},remember picture]
    
\begin{scope}[shift={(14,2)}, , scale=0.23,every node/.append style={scale=0.8}]
        \fill[gray!30] (0,0) -- (24,0) -- (24,1) -- (0,1) -- cycle;
    \fill[gray!30] (0,0) -- (1,0) -- (1,10) -- (0,10) -- cycle;
    \fill[gray!30] (0,9) -- (24,9) -- (24,10) -- (0,10) -- cycle;

    \fill[orange!30] (4,3) rectangle (8,7);
    \fill[orange!30] (12,3) rectangle (16,7);
    \fill[orange!30] (20,3) rectangle (24,7);

    \draw[->, thick] (6,9) -- (6,7) node[right,midway] {$x$};
    \draw[->, thick] (14,9) -- (14,7)node[right,midway] {$y$};
    \draw[->, thick] (22,9) -- (22,7) node[right,midway] {$z$};
    
    \draw[->, thick] (6,3) -- (6,1)node[right,midway] {$a$};
    \draw[->, thick] (14,3) -- (14,1)node[right,midway] {$b$};
    \draw[->, thick] (22,3) -- (22,1)node[right,midway] {$c$};

    \fill[orange!30] (2,11) -- (20,11) -- (20,13) -- (2,13) -- cycle;
    \fill[orange!30] (2,3) -- (4,3) -- (4,11) -- (2,11) -- cycle;
    \fill[orange!30] (10,3) -- (12,3) -- (12,11) -- (10,11) -- cycle;
    \fill[orange!30] (18,3) -- (20,3) -- (20,11) -- (18,11) -- cycle;

\end{scope}

\draw[->, thick] (10.2,3.2) -- (13,3.2)node[above,midway] {$n \to \infty$};

\node at (11.6,4) {\large{Sequential NPA}};
\node (LabelSectionNPA) at (11.6, 2.4) {};

\node at (16.8,1) {\large{Commuting operator strategies}};

\node at (8,-1) {\Large $\text{ONS}^n =0$};

\begin{scope}[shift={(7,0)}, , scale=0.23,every node/.append style={scale=0.8}]

\fill[gray!30] (8.5,0) -- (9.5,0) --(9.5,28) -- (8.5,28) -- cycle;
    \fill[gray!30] (3,0) -- (8.5,0) --(8.5,1) -- (3,1) -- cycle;
    \fill[gray!30] (3,9) -- (8.5,9) --(8.5,10) -- (3,10) -- cycle;
    \fill[gray!30] (3,18) -- (8.5,18) --(8.5,19) -- (3,19) -- cycle;
    \fill[gray!30] (3,27) -- (8.5,27) --(8.5,28) -- (3,28) -- cycle;

    \fill[orange!30] (0,3) rectangle (6,7);
    \fill[orange!30] (0,12) rectangle (6,16);
    \fill[orange!30] (0,21) rectangle (6,25);

    \fill[orange!60] (0.5,12) rectangle (1.5,8);
    \node[isosceles triangle,
	isosceles triangle apex angle=90,
	fill=orange!60,
        rotate=-90,
	minimum size =5] (T90)at (1,7.8){};

        \fill[orange!60] (0.5,21) rectangle (1.5,17);
    \node[isosceles triangle,
	isosceles triangle apex angle=90,
	fill=orange!60,
        rotate=-90,
	minimum size =5] (T90)at (1,16.8){};



    \draw[->, thick] (4,27) -- (4,25) node[right,midway] {$x$};
    \draw[->, thick] (4,18) -- (4,16)node[right,midway] {$y$};
    \draw[->, thick] (4,9) -- (4,7) node[right,midway] {$z$};
    
    \draw[->, thick] (4,21) -- (4,19)node[right,midway] {$a$};
    \draw[->, thick] (4,12) -- (4,10)node[right,midway] {$b$};
    \draw[->, thick] (4,3) -- (4,1)node[right,midway] {$c$};


	isosceles triangle apex angle=90,
	fill=orange!30,
        rotate=-90,
	minimum size =5] (T90)at (6,0){};
\end{scope}

\node at (5,3.2) {\LARGE$\approx_{\negl(\lambda)}$};

\node at (4.5,4) {\large{Geometric argument}};
\node (LabelSectionSoundness) at (4.5,2.4) {};

\node at (1.5,-1) {\Large $\text{ONS}^n\leq \negl(\lambda)$};

\begin{scope}[shift={(0,0)}, scale=0.23,every node/.append style={scale=0.8}]
    \fill[gray!30] (8.4,0) -- (9.4,0) --(9.4,28) -- (8.4,28) -- cycle;
    \fill[gray!30] (3,0) -- (8.5,0) --(8.5,1) -- (3,1) -- cycle;
    \fill[gray!30] (3,9) -- (8.5,9) --(8.5,10) -- (3,10) -- cycle;
    \fill[gray!30] (3,18) -- (8.5,18) --(8.5,19) -- (3,19) -- cycle;
    \fill[gray!30] (3,27) -- (8.5,27) --(8.5,28) -- (3,28) -- cycle;

    \fill[orange!30] (2,3) rectangle (6,7);
    \fill[orange!30] (2,12) rectangle (6,16);
    \fill[orange!30] (2,21) rectangle (6,25);

    \fill[orange!30] (0,3) -- (2,3) -- (2,25) -- (0,25) -- cycle;


    \draw[->, thick] (4,27) -- (4,25) node[right,midway] {$\mathsf{Enc}(x)$};
    \draw[->, thick] (4,18) -- (4,16)node[right,midway] {$\mathsf{Enc}(y)$};
    \draw[->, thick] (4,9) -- (4,7) node[right,midway] {$z$};
    
    \draw[->, thick] (4,21) -- (4,19)node[right,midway] {$\mathsf{Enc}(a)$};
    \draw[->, thick] (4,12) -- (4,10)node[right,midway] {$\mathsf{Enc}(b)$};
    \draw[->, thick] (4,3) -- (4,1)node[right,midway] {$c$};
\end{scope}

\end{tikzpicture}
%
%
\begin{tikzpicture}[remember picture,overlay]
  \node[anchor=center] at (LabelSectionSoundness.center) {\small\strut\cref{sec:SoundnessCompileMain}};
  \node[anchor=center] at (LabelSectionNPA.center) {\small\strut\cref{sec:SeqNPAMain}};
\end{tikzpicture}}
    \caption{A graphical overview of our two main technical contributions: the sequential NPA hierarchy and a geometric argument for closeness of compiled strategies. (Middle to Right): Our sequential NPA hierarchy characterizes ideal sequential quantum strategies (middle) where operational non-signaling holds perfectly ($\text{ONS}^n=0$). We show that its solutions converge to the set of standard commuting operator strategies (right) as $n \to \infty$. (Left to Middle): A compiled strategy (left) is not a perfect solution but an approximate one to our sequential hierarchy, violating the non-signaling condition negligibly ($\text{ONS}^n\leq \negl(\lambda)$). We provide a novel geometric argument which proves that this pseudo-solution is negligibly close to a genuinely feasible solution of our hierarchy.}
    \label{fig:scheme}
\end{figure}

\vspace{0.3em}
\noindent
\textbf{Overcome the limitations.}
This paper overcomes both obstacles by adopting the same high-level approach while introducing entirely new, scalable tools, as demonstrated by \cref{fig:scheme}.
First, to replace the unextendable bipartite model, we introduce a composable multipartite sequential NPA hierarchy based on the Heisenberg picture, which is guaranteed to converge for any number of players.
Second, to replace the non-scalable proof technique, we develop a novel, concise, and scalable geometric proof (further illustrated by \cref{fig:GeometryOfProof}).
This new argument establishes the crucial closeness result between a near-solution from a compiled strategy and a true feasible solution in our hierarchy.
Together, these contributions prove quantitative soundness for compiled games with any number of players.
We now detail our techniques in two parts.

\vspace{0.3em}
\noindent\textbf{The NPA hierarchy for nonlocal settings (\cref{fig:nonlocal-a}).}
The Navascu\'es-Pironio-Ac\'in (NPA) hierarchy~\cite{navascues2008convergent, pironio2010convergent}, a noncommutative generalization of the Lasserre-Parrilo hierarchy~\cite{lasserre2001global, parrilo2003semidefinite}, is a cornerstone for analyzing nonlocal games.
It provides a sequence of increasingly tight semidefinite programming (SDP) relaxations, indexed by an integer level $n$, that systematically computes upper bounds on the optimal quantum score.

In particular, the standard NPA hierarchy is designed for the standard Bell scenario (\cref{fig:nonlocal-a}), where all players' actions are modeled as terminal measurement POVMs that commute.
Specifically, consider a game $\cG$ with questions $x, y, z, \dots$ and answers $a, b, c \dots$, the POVMs of all players are modeled with commuting letters $\{f_{a|x}\}, \{f_{b|y}\}, \{f_{c|z}\}, \dots$, and denote the shared quantum state by $\sigma$.
These letters form operator words (monomials); the number of letters in a word is its degree.

The core of the NPA technique is a level-$n$ matrix, $\Gamma^{(n)}$, indexed by operator words of degree $\leq n$.
The entries of this matrix correspond to the expectation values of the associated operator words with respect to $\sigma$.
That is, for any two words $w$ and $v$ of degree $\leq n$, the corresponding matrix entry can be understood as:
\begin{align*}
    \Gamma^{(n)}_{w,v} = L^{2n}(w^*v) = \Tr{\sigma \cdot w^*v},
\end{align*}
where $L^{2n}$ is the associated linear map modeling the expectation $\Tr{\sigma \ \cdot}$.
As $n \to \infty$, the hierarchy of matrices encodes all possible expectation values and thus completely characterizes the set of commuting observable quantum strategies.

However, this framework cannot describe sequential protocols (\cref{fig:nonlocal-b}) or their compiled counterparts (\cref{fig:nonlocal-c}), which are central to our work, where the actions of earlier players are quantum instruments that transform the state for subsequent players.
This necessitates a fundamental generalization of the NPA framework for sequential settings. Some partial results towards the construction of the NPA hierarchy for restricted sequential setting exist in the literature~\cite{bowles2020bounding,tavakoli2021bounding,chaturvedi2021characterising}, but without convergence guarantees, and inapplicable to general sequential scenarios.

\vspace{0.3em}
\noindent\textbf{Our contribution: A multipartite sequential NPA hierarchy for quantum instruments (\cref{fig:nonlocal-b}).}
To this end, we introduce a composable, multipartite sequential NPA hierarchy (\cref{eq:kpartite_SequentialNPASDP}) and prove that this hierarchy is complete for both the quantum sequential and nonlocal scenarios (\cref{thm:MainCompleteness}).
We now sketch its construction, focusing on the tripartite case for demonstration purposes.

We use notations of \cref{fig:nonlocal-b}, based on the Heisenberg algebraic picture of a sequential game, dual to the standard Schr\"odinger's picture in which maps act on states.
More precisely, the quantum instruments modeling Alice's and Bob's actions do not act on the state $\sigma$ which is fixed, but on Charlie's measurement operators (see \cref{eq:SequentialityDemonstration} below).
Alice's and Bob's actions are described by completely positive (CP) maps (quantum instruments) $\{\QInstrument_{a|x}\}_a$ and $\{\QInstrument_{b|y}\}_b$, while the final player, Charlie, performs a standard measurement with POVM effects $\{f_{c|z}\}_c$.
With these objects, the corresponding correlations are
\begin{align*}
    p(abc|xyz)=\Tr{\sigma \cdot \QInstrument_{a|x} \circ \QInstrument_{b|y}(f_{c|z})},
\end{align*}
achieving the score
\begin{align*}
    \sum_{a,b,x,y} \beta_{abcxyz} p(abc|xyz),
\end{align*}
where $\beta_{abcxyz}$ is the payoff tensor associated with the rule of the game $\cG$.
Diagrammatically, the sequentiality can be represented by
\begin{align}\label{eq:SequentialityDemonstration}
    f_{c|z} \text{ of } C \xrightarrow{\QInstrument_{b|y}} \QInstrument_{b|y}(f_{c|z}) \text{ of } B \xrightarrow{\QInstrument_{a|x}} \QInstrument_{a|x} \circ \QInstrument_{b|y}(f_{c|z}) \text{ of } A.
\end{align}
The final pieces of the sequential scenarios are the \emph{operationally-non-signaling constraints} of Alice's and Bob's actions, which are, respectively,
\begin{align*}
    \sum_a \QInstrument_{a|x}=\sum_a \QInstrument_{a|x'}, \quad \sum_b \QInstrument_{b|y}=\sum_b \QInstrument_{b|y'},
\end{align*}
for all $x, x', y, y'$.

Our sequential generalization of the NPA hierarchy must therefore construct moment matrices $\Gamma^{(n)}$ that algebraically encode both the ordered action of these instruments and their CP-map-level operationally-non-signaling constraints.
The fundamental challenge is twofold: (1) to define a new set of operator words that captures the sequential structure of the game, and (2) to impose the appropriate algebraic and positivity constraints on the moment matrix $\Gamma^{(n)}$ (equivalently, on its associated linear map $L^{2n}$) to faithfully model the sequential scenarios.

For (1), the solution is a sequential construction of the ``words'' (operator monomials) that index our $\Gamma^{(n)}$, which we start from the last player as in \cref{eq:SequentialityDemonstration}:
\begin{enumerate}[label=(\alph*), noitemsep, topsep=2pt, leftmargin=*]
    \item \emph{The final player Charlie}: We begin with Charlie, whose action is a standard POVM.
    Similarly to the standard NPA hierarchy, we describe his measurements with a set of operator letters $\{f_{c|z}\}$.
    These letters form words (monomials) such as $u = f_{c_1|z_1} f_{c_2|z_2} \cdots f_{c_m|z_m}$, which have a degree $m$, denoted by $\deg u = m$.
    The set of all such words with $\deg u \leq n$ forms our base word set, $\cW^n_{C}$.

    \item \emph{The preceding player Bob}: 
    To model Bob's quantum instrument $\{\QInstrument_{b|y}\}$, we consider its Stinespring dilation, denoted ${T_{b|y}}$.
    This choice is physically motivated by the Stinespring dilation theorem, which states that any such instrument (a completely positive map) can be realized by a unitary evolution on a larger system.

    Crucially, a unitary evolution preserves the algebraic structure of operators.
    This means our model maps, ${T_{b|y}}$, are required to be *-homomorphisms, satisfying:
    \begin{align*}
        T_{b|y}(f_{c|z}^*) = T_{b|y}(f_{c|z})^* \quad \text{and} \quad T_{b|y}(f_{c|z} f_{c'|z'}) = T_{b|y}(f_{c|z}) T_{b|y}(f_{c'|z'})
    \end{align*}
    for all of Charlie's measurement operators $f_{c|z}, f_{c'|z'}$ and consequently for all words from $\cW^n_{C}$.
    These maps act on Charlie's operators to generate new letters, $f_{bc|yz} := T_{b|y}(f_{c|z})$, which algebraically capture the effect of Bob's instrument $\QInstrument_{b|y}$ on the subsequent system.
    The set of all words formed from these new letters $f_{bc|yz}$ up to degree $n$ is denoted $\cW^n_{BC}$.

    \item \emph{The first player Alice}: We proceed sequentially for Alice, defining her corresponding Stinespring dilated $*$-homomorphisms $\{T_{a|x}\}$ which act on the words of Bob and Charlie.
    This generates the final set of letters $f_{abc|xyz} = T_{a|x}( f_{bc|yz} )$ 
    and the degree $\leq n$ word set $\cW_{ABC}^n$.
    This layered construction directly models the sequentiality of the protocol analogous to \cref{eq:SequentialityDemonstration}:
    \begin{align*}
        f_{c|z} \text{ (Charlie)} \xrightarrow{T_{b|y}} f_{bc|yz} \text{ (Bob)} \xrightarrow{T_{a|x}} f_{abc|xyz} \text{ (Alice)}.
    \end{align*}

    \item \emph{The moment matrix}: This sequential process defines the full set of words $\cW_{ABC}^n$ that index our moment matrix $\Gamma^{(n)}$.
    The entry $\Gamma^{(n)}_{\id, f_{abc|xyz}}$ represents the correlation $p(abc|xyz)$, and consequently corresponding score in game $\cG$ with the payoff tensor $\beta_{abcxyz}$ is
    \begin{align*}
        \sum_{a,b,c,x,y,z} \beta_{abcxyz} \Gamma^{(n)}_{\id, f_{abc|xyz}}.
    \end{align*}
\end{enumerate}

With the indices of $\Gamma^{(n)}$ sequentially defined, we now address (2) by identifying the appropriate constraints to recover a sequential quantum strategy as the NPA level $n \to \infty$.
\begin{enumerate}[label=(\Alph*), noitemsep, topsep=2pt, leftmargin=*]
    \item \emph{Standard constraints}: 
    Since $\Gamma^{(n)}$ and $L^{2n}$ are to model the expectation map $\Tr{\sigma \  \cdot}$, this physical interpretation leads directly to three fundamental constraints on $\Gamma^{(n)}$.
    First, for $L^{2n}$ to be a well-defined linear map on the word sets, the matrix $\Gamma^{(n)}$ must be symmetric, and satisfy a Hankel-like condition.
    Second, since $\Tr{\sigma \  \cdot}$ is a positive map for any state $\sigma$, the functional $L^{2n}$ must also be positive.
    This is precisely equivalent to the constraint that $\Gamma^{(n)}$ must be positive semidefinite (PSD).
    Third, the state normalization $\Tr{\sigma} = 1$ corresponds to $L^{2n}(\id) = 1$, which implies the matrix normalization condition $\Gamma^{(n)}_{\id, \id} = 1$.
    
    \item \emph{Alice's operationally-non-signaling constraint}:
    A crucial innovation of our hierarchy are \emph{CP map level constraints} that formalize the notion of operationally-non-signaling, meaning that the marginal maps $\sum_a T_{a|x}$ and $\sum_a T_{a|x'}$ are the same if only tested against polynomials up to degree $n$; as $n$ goes to infinity, perfect operationally-non-signaling constraints are retrieved.
    The constraints for Alice's instrument read:
    \begin{align*}
        L^{2n}\Bigl(w^* \sum_a \bigl( T_{\ax}(r^*s) - T_{a|x'}(r^*s) \bigr) v \Bigr) = 0 
    \end{align*}
    for all $x, x'$, and for all words $r, s$ built from $\{f_{bc|yz}\}$ and $w,v$ from $\{f_{abc|xyz}\}$ up to a total degree $\leq 2n$.
    Here, the central term $\sum_a \bigl( T_{\ax}(r^*s) - T_{a|x'}(r^*s) \bigr)$ tests Alice's non-signaling condition on an arbitrary operator $r^*s$ from Bob's algebra.
    This test is then embedded within the context of arbitrary actions from the full sequential strategy that Alice can perform (words $w, v$), ensuring the condition holds universally.
    
    \item The constraint for Bob's instrument is more involved, but demonstrates the composable structure of our method:
    \begin{align*} 
      L^{2n}\!\Bigl(
          T_{a|x}\Bigl(r^* \sum_{b}\bigl(T_{b|y}(t^{*}u)-T_{b|y'}(t^{*}u)\bigr) s\Bigr)
      \Bigr) = 0,
    \end{align*}
    for all $a,x,y,y'$, and for all words $r,s$ built from $\{f_{bc|yz}\}$ and $t,u$ from $\{f_{c|z}\}$ up to a total degree $\leq 2n$.
    Similarly, the central term, $\sum_{b}\bigl(T_{b|y}(t^{*}u)-T_{b|y'}(t^{*}u)\bigr)$, tests Bob's non-signaling condition on an arbitrary operator $t^*u$ from Charlie's algebra.
    This core expression is embedded within the context of arbitrary operators from Bob's layer (words $r, s$) and placed under every possible action by Alice ($T_{a|x}, \, \forall a,x$).
    In this way, the SDP enforces Bob's non-signaling constraints across every possible algebraic context created by the preceding players.
\end{enumerate}

Following (a)-(d) and (A)-(C), we reach the subsequent definition of the tripartite sequential NPA hierarchy that maximizes the score achievable by all sequential quantum strategies (\cref{eq:Tripartite_SequentialNPASDP}):
\begin{align}\label{eq:Intro_Tripartite_SequentialNPASDP}
    \gamevalueTriSeqNPA{\cG}{n}
    \;:=\;
    \max\quad
      & \sum_{a,b,c,x,y,z} \beta_{abcxyz}\;
        \Gamma^{(n)}_{\id, f_{abc|xyz}} \nonumber \\
    \text{s.t.}\quad
      &\Gamma^{(n)} \succeq 0, \quad
      \Gamma^{(n)}_{\id,\id}=1, \quad\text{(PSD and normalization)} \nonumber \\
    &\Gamma^{(n)}_{w,w'}=\Gamma^{(n)}_{v,v'}
    \quad\text{whenever }w^{*}w'=v^{*}v', \quad \text{(Hankel condition)} \nonumber \\[2pt]
    &L^{2n}\Bigl(w^* \sum_a \bigl( T_{\ax}(r^*s) - T_{a|x'}(r^*s) \bigr) v \Bigr)\;= 0 \nonumber \\[2pt]
    &\hphantom{M_{(w,r),(s,v)}:=}
    \substack{\forall x, x', \, \forall w,v\in\cW_{ABC}^{n},\\
              r,s\in\cW_{BC}^{n},\\
              \deg w+\deg v+\deg r+\deg s\le 2n}
    \quad\text{(Alice operationally‑non‑signaling)} \nonumber \\[2pt]
    &L^{2n}\!\Bigl(
      T_{a|x}\Bigl(r^* \sum_{b}\bigl(T_{b|y}(t^{*}u)-T_{b|y'}(t^{*}u)\bigr) s\Bigr)
    \Bigr) \;= 0 \nonumber \\[2pt]
    &\hphantom{N_{(r,t),(s,u)}:=}
    \substack{\forall\, a,x, y, y', 
    \\ r,s\in\cW_{BC}^{n},\\
              t,u\in\cW_{C}^{n},\\
              \deg r+\deg s+\deg t+\deg u\le 2n}
    \quad\text{(Bob operationally‑non‑signaling)}.
\end{align}

As showcased in the sequential construction (a)-(d), it is straightforward to obtain the word sets to any $k$-partite scenarios.
Analogously, the operationally-non-signaling constraints in (B) and (C) can be generalized to $k$-partite scenarios in an inductive way. 
Indeed, by induction, we formulate in \cref{eq:kpartite_SequentialNPASDP} the general $k$-partite sequential NPA hierarchy. 
As stated in \cref{thm:MainCompleteness}, this hierarchy is complete to both the standard multipartite quantum strategies (\cref{fig:NonlocalCompiledBellGame}(a)) and the multipartite quantum sequential strategies (\cref{fig:NonlocalCompiledBellGame}(b)).

In contrast to the bipartite hierarchies of~\cite{klep2025quantitative,cui2025convergent} formulated in the Schr\"odinger picture, our novel sequential construction operates in the Heisenberg picture.
Instead of tracking the evolving state, we model the transformations themselves using $*$-homomorphisms and include all off-diagonal terms that were previously discarded in the bipartite sequential hierarchies (which are block-diagonal in $(a,x)$).
This preserves the complete algebraic structure of each quantum instrument in the sequence.
This perspective provides a clean and rigorous way to encode the CP-map-level non-signaling constraints, leading to a novel hierarchy that is powerful enough to converge for any number of players (\cref{thm:MainCompleteness}).
We expect this technique to be of independent interest for the computational analysis of multi-round quantum protocols and in fields where CP maps are fundamental~\cite{paulsen2002completely,raginsky2003radon}.

\vspace{0.3em}
\noindent\textbf{Relating compiled nonlocal game (\cref{fig:nonlocal-c}) to the sequential NPA hierarchy with a scalable geometric argument.}
We now explain the final piece of our hierarchy-based quantitative soundness bound: relating compiled games to our new sequential hierarchy as shown in the left part of \cref{fig:scheme}.
The high-level strategy, mirroring that of~\cite{klep2025quantitative}, is to show that any efficient compiled strategy is ``close'' to a genuinely feasible solution within our multipartite sequential NPA hierarchy.
However, as previously discussed, the representation-based signaling decomposition argument used in~\cite{klep2025quantitative} and the pseudo-expectation lemma of~\cite{cui2025convergent} are intrinsically bipartite and fail here.

\begin{figure}
    \centering
    \includegraphics[width=0.35\linewidth]{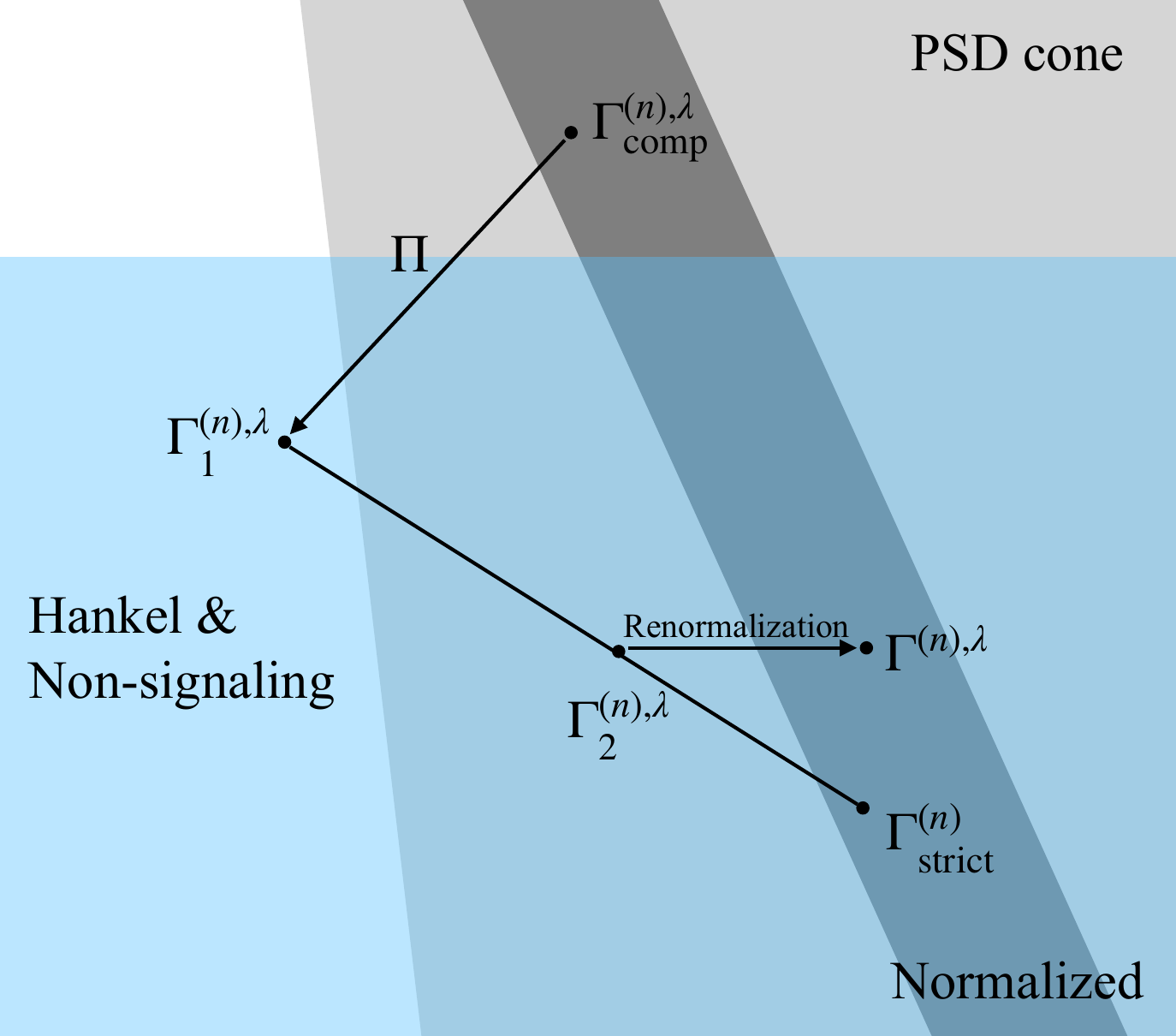}
    \caption{Geometric sketch of proof for \cref{thm:MainQuantumSoundness}, detailed in \cref{sec:GeometricProof}. The gray region is the PSD cone; the blue slice encodes the Hankel condition and operationally-non-signaling constraints of our $k$-partite sequential NPA hierarchy \cref{eq:kpartite_SequentialNPASDP}. The thick dark line indicates normalized moment matrices. From a compiled strategy we extract $\Gamma^{(n),\lambda}_{\comp}$. Applying the projector $\Pi$ (constructed based on \cref{eq:EMapForConstraint}) yields $\Gamma^{(n), \lambda}_1$ that satisfies the affine constraints but may fail PSD and normalization. Positivity is restored by convexly mixing with a strictly feasible point $\Gamma^{(n)}_\strict$, giving $\Gamma^{(n), \lambda}_2$, which is then rescaled to a normalized $\Gamma^{(n), \lambda}$ that is a feasible solution to our multipartite sequential NPA hierarchy. We are done by invoking \cref{thm:MainCompleteness}.}
    \label{fig:GeometryOfProof}
\end{figure}

We therefore introduce a novel, scalable geometric proof that achieves this closeness result (illustrated in \cref{fig:GeometryOfProof} and fully presented in \cref{sec:SoundnessCompileMain}).
The core idea is as follows:
\begin{enumerate}[noitemsep, topsep=2pt, leftmargin=*]
    \item As a consequence of the block-encoding arguments introduced in~\cite{natarajan2023bounding,kulpe2024bound} and further developed in~\cite{baroni2025asymptotic}, an efficient compiled strategy, with a finite cryptographic security $\lambda$, produces a moment matrix $\Gamma^{(n),\lambda}_{\comp}$ that acts as a ``pseudo-solution'' for our hierarchy.
    It is positive and normalized, but it weakly violates the operationally-non-signaling constraints.
    Based on $\Gamma^{(n),\lambda}_{\comp}$, our proof then geometrically constructs a nearby, genuinely feasible solution.
    \item
    We first project $\Gamma^{(n),\lambda}_{\comp}$ onto the affine subspace of matrices that perfectly satisfy the non-signaling and Hankel constraints. However, this projected matrix might no longer be PSD and normalized.
    \item 
    We then restore positivity by taking a slight convex combination with a known, strictly feasible solution $\Gamma^{(n)}_{\strict}$ (\cref{prop:kpartite_StrictFeasible}).
    Finally, we re-normalize the result to obtain a genuine, feasible solution $\Gamma^{(n),\lambda}$ that satisfies all constraints of our hierarchy.
    \item We prove that $\Gamma^{(n),\lambda}$ is negligibly close to the original $\Gamma^{(n),\lambda}_{\comp}$ in operator norm (\cref{thm:kpartite_BoundToFeasibleSolution}).
    \item By combining the above closeness statement with \cref{thm:MainCompleteness}, we obtain the hierarchy-based quantitative upper bound in \cref{thm:MainQuantumSoundness}; finite-level convergence and flat optimality then give the sharpened soundness corollaries.
\end{enumerate}
This geometric approach provides a powerful template for establishing quantitative control in complex quantum scenarios where representation-based algebraic arguments are insufficient.

The key advantage is its scalability and generality.
For example, we note that our geometric proof technique can also be used to analyze the convergence of the recently introduced computational-SoC hierarchy~\cite[Definition~4.3]{merkulov2025computationalbellinequalities}.
Specifically, our argument implies that any optimal solution to their hierarchy must be negligibly close to a feasible solution of the bipartite sequential NPA hierarchy from~\cite{klep2025quantitative}.

\subsection{Further discussions}\label{sec:IntroDiscussion}
Our results introduce a powerful convergent sequential NPA-like hierarchy for analyzing sequential quantum protocols.
With this new hierarchy, we provide upper bounds on the compiled scores of all multipartite nonlocal games and, under finite-level hierarchy certificates, obtain quantitative quantum soundness with respect to the relevant quantum value.
We briefly discuss the implications of both contributions.

\vspace{0.3em}
\noindent
\textbf{Resolving quantum soundness for multipartite games with a finite-level sequential NPA hierarchy certificate.}
Consider a multipartite nonlocal game $\cG$ for which the sequential hierarchy is certified to collapse at some level $n_0$. 
Then \cref{thm:MainQuantumSoundness} yields a quantitative soundness guarantee for the compiled game with respect to $\omega_{\mathrm{qc}}(\cG)$.
Under the stronger flat-optimality assumption, this becomes a guarantee with respect to the tensor-product quantum value $\omega_{\mathrm q}(\cG)$, and the flat solution also yields a finite-dimensional optimal strategy.

Importantly, once the level $n_0$ is known, the resulting negligible term can be made explicit without solving higher levels of the sequential hierarchy.
Indeed, for any fixed level $n$, the function $\negl_{S,n}(\lambda)$ depends only on the compiled strategy $S$, the underlying QHE scheme, and the fixed game syntax, since it is obtained by taking the maximum QHE-security loss over all monomials of length at most $n$, multiplied by explicit $n$-dependent constants.

\vspace{0.3em}
\noindent
\textbf{Other multipartite games.}
For other multipartite games, \cref{eq:MainthmAllGameUpperBound} provides a concrete bound on the maximal score obtainable by an efficient compiled prover.
To use this bound in practice one needs to solve the sequential NPA hierarchy at some level $n$.
In this general setting we do not claim a fully quantitative soundness statement with respect to the true quantum value, since this would require knowing (or certifying) a level at which the hierarchy is tight.

However, the question of whether such a quantitative statement can be achieved in principle is subtle.
Indeed, for games whose optimal value requires very large (or even infinite) dimension, playing an optimal strategy may already be out of reach even in the non-compiled setting for spacelike separated provers.

\vspace{0.3em}
\noindent
\textbf{Future directions for the multipartite sequential NPA hierarchy.}
Beyond its immediate application in this work, our multipartite sequential NPA hierarchy for quantum instruments is a technical contribution that could be of independent interest.
It provides the first systematic, convergent, and composable tool for computing bounds on the capabilities of sequential quantum strategies.
As such, it naturally complements the recent multipartite generalization of the S-G-HJW purification theorem~\cite{baroni2025asymptotic}.
This opens up several avenues for future research:
\begin{enumerate}[noitemsep, topsep=2pt, leftmargin=*]
    \item \emph{Numerical performance and convergence.}
    Our sequential NPA hierarchy involves smaller moment matrices at each level compared to the standard NPA hierarchy, but incorporates more complex constraints.
    Understanding this trade-off between the rate of convergence and the moment matrix size can be interesting both theoretically and important for practical numerical implementation.
    Such insight can have further implications to related SDP hierarchies, such as the sparse SOS hierarchies~\cite{klep2022sparse,magron2023sparse} and the bipartite sequential NPA hierarchies~\cite{klep2025quantitative}.
    \item \emph{Extending device-independent certification.} The standard NPA hierarchy is often the foundation for the device-independent certification of properties like randomness and entanglement.
    Our hierarchy extends this capability to a broader class of sequential and multi-round protocols, enabling the certification of tasks where players act one after another.
    \item \emph{Beyond operationally-non-signaling constraints.} 
    The algebraic nature of our framework is not limited to non-signaling.
    It can be adapted to enforce other constraints on quantum instruments to describe different scenarios.
    For instance, one could require a specific instrument outcome ($a=0$) to correspond to an identity channel, a common feature in error correction or gate-based protocols.
    Moreover, it can be adapted to model new constraints due to cryptographic primitives beyond the homomorphic encryption.
    
    \item \emph{Analyzing quantum interactive proofs.}
    Our current work models a linear sequence of players.
    However, its composable, algebraic nature suggests it could be extended to analyze protocols with more complex interaction structures.
    It could provide a framework for understanding more general interactive quantum protocols (where a quantum prover with memory interacts with a verifier over multiple rounds), a central scenario in quantum complexity theory and quantum cryptography.
\end{enumerate}

In summary, our work provides a general hierarchy-based approach to quantum soundness and a new analytical tool for studying compiled strategies, with direct implications for single-device quantum cryptography offers new capabilities for the broader study of quantum information protocols.

\subsection{Outline of the manuscript}\label{sec:IntroOutline}
The remainder of this paper is organized as follows.
\cref{sec:Preliminaries} presents necessary preliminaries: We review the necessary background, defining nonlocal and compiled games (\cref{sec:PrelimNonlocalGameCompiledGame}), 
and the standard NPA hierarchy (\cref{sec:PrelimNPAHierarchy}).
\cref{sec:SeqNPAMain} then presents the novel multipartite sequential NPA hierarchy, where we introduce the construction progressively: \cref{sec:SeqNPABiparititeCase} first revisits the bipartite case to provide a proof of concept and contrast our Heisenberg-picture approach with the Schr\"odinger-picture method of~\cite{klep2025quantitative}.
\cref{sec:SeqNPATripartiteCase} then uses the tripartite case as a pedagogical bridge to the general construction.
\cref{sec:SeqNPAMultipartiteCase} presents the full, general $k$-partite hierarchy and proves its key properties: completeness, the flatness condition, and strict feasibility (\cref{thm:MainCompleteness}).
\cref{sec:SoundnessCompileMain} proves our main quantitative soundness theorem (\cref{thm:MainQuantumSoundness}).
After establishing notation in \cref{sec:SoundnessCompileNotation}, the argument is based on a novel geometric proof developed in \cref{sec:GeometricProof}.
This proof, using projection, regularization, and normalization, establishes the closeness result (\cref{thm:kpartite_BoundToFeasibleSolution}).
Finally, \cref{sec:PiecesTogetherCorollary} combines this closeness result and \cref{thm:MainCompleteness} to complete the proof of the main theorem.
%
The manuscript finishes with \cref{sec:Appendix} on a brief comparison on~\cite{klep2025quantitative,baroni2025asymptotic}.

We refer to \cref{fig:scheme} for a graphical illustration of \cref{sec:SeqNPAMain,sec:SoundnessCompileMain}, and \cref{fig:big-scheme} for a more comprehensive graphical mind map of the paper.

\section{Preliminaries}\label{sec:Preliminaries}
We introduce the notations of nonlocal games and compiled games in \cref{sec:PrelimNonlocalGameCompiledGame}, followed by a brief introduction of the NPA hierarchy in \cref{sec:PrelimNPAHierarchy}.

\subsection{Nonlocal games and compiled games}\label{sec:PrelimNonlocalGameCompiledGame}
A nonlocal game $\cG$ is an interaction involving a referee and multiple players (provers) who cannot communicate with each other during the game; see \cref{fig:NonlocalCompiledBellGame}(a) for a tripartite example. The referee sends each player a question $x_i$ (drawn from some specified distribution $\mu(\vec{x})$), and each player must respond with an answer $a_i$. Whether the players win is decided by a publicly known rule (a predicate depending on all questions and answers) $V(\vec{a},\vec{x})$. No communication means each player must base their answer only on their own question and a prior agreed-upon strategy among the players, but not on the other players' questions or answers. The players' goal is to maximize their winning probability $\sum_{\vec{x}\vec{a}}\mu(\vec{x})V(\vec{a},\vec{x})p(\vec{a}|\vec{x})$ by coordinating a strategy ahead of time (they know the game's description in advance).

The \emph{classical value} (or score) of a nonlocal game is the maximum success probability achievable when the players share only classical resources, e.g., a pre-shared random string (common randomness) or any predetermined classical strategy. Denoting with $f_i$ an arbitrary deterministic function mapping input $x_i$ to output $a_i$, the classical value takes the form
\begin{align*}
    \omega_c(\cG) = \max_{f_1,\cdots,f_k}\sum_{\vec{x}}\mu(\vec{x})V(x_1,\cdots,x_k,f_1(x_1),\cdots,f_k(x_k)).
\end{align*}
In contrast, the \emph{quantum value} (or score) is the maximum winning probability when players can use quantum resources such as shared entangled states $\rho$ and quantum measurements $M_{a_i|x_i}$:
\begin{align*}
    \omega_q(\cG) = \sup_{\rho,M_{a_1|x_1},\cdots,M_{a_k|x_k}}\sum_{\vec{a},\vec{x}}\mu(\vec{x})V(\vec{x},\vec{a})\Tr{\bigotimes_{i=1}^kM_{a_i|x_i}\rho}.
\end{align*}
Quantum strategies can outperform classical ones in certain games, a phenomenon known as a nonlocal \emph{quantum advantage}. For example, in the famous CHSH game~\cite{clauser1969proposed}, the best classical strategy wins with probability $0.75$, whereas players sharing an entangled qubit pair can win with probability about $0.85$. This higher success rate ($\approx85\%$ vs $75\%$) is a quantum advantage that cannot be achieved by any classical means. Such nonlocal games (also called Bell games in physics) thus highlight the stronger-than-classical correlations allowed by entanglement.

It is worth noting a subtle point about the definition of ``quantum value''. In finite-dimensional quantum systems, insisting that players cannot communicate is equivalent to requiring that they act on separate Hilbert spaces (tensor factors) of a joint entangled state. 
More generally, one can impose that all of one player's measurement operators commute with all of the other player's operators, which is a formal way to enforce a non-signaling condition in possibly infinite-dimensional or arbitrary systems. 
While the tensor-product and commuting definitions coincide for finite dimensions~\cite{scholz2008tsirelson}, a result now known to be robust~\cite{ozawa2013tsirelson,xu2025quantitative}, they differ in general~\cite{ji2022mipre}. 
This leads to the definition of the \emph{quantum commuting operator value} of a game, which is the supremum win probability attained by any strategy where the players' operations commute:
\begin{align*}
    \omega_{\mathrm{qc}}(\cG) = \sup_{\rho,N_{a_1|x_1},\cdots,N_{a_k|x_k}}\sum_{\vec{a},\vec{x}}\mu(\vec{x})V(\vec{x},\vec{a})\Tr{\prod_{i=1}^k N_{a_i|x_i}\rho}.
\end{align*}

While nonlocal games traditionally require multiple spatially isolated devices, the recent line of work initiated by \cite{kalai2023quantum} shows that, under computational assumptions, one can compile a $k$-prover game $\cG$ to produce a new protocol between a single QPT prover and a classical verifier, which we call a compiled nonlocal game.
Intuitively, the compilation procedure simulates space-like separation using cryptography and a specific sequential structure as in \cref{fig:NonlocalCompiledBellGame}(c): the verifier sends an encrypted question, and always wait for an encrypted answer before sending the following question; the last question-answer round doesn't need to be encrypted.
\begin{definition}[Compiled $k$-partite nonlocal game]\label{def:CompiledGame}
    Let $\cG$ be a $k$-partite nonlocal game with input alphabets $X_i$, answer alphabets $A_i$, input distribution $\mu$ on $X := X_1 \times \cdots \times X_k$, and predicate $V(\vec{x},\vec{a})$, where $\vec{x}=(x_1,\ldots,x_k)$ and $\vec{a}=(a_1,\ldots,a_k)$.
    Fix a quantum homomorphic encryption scheme $\mathsf{QHE} = (\Gen, \Enc, \Eval, \Dec).$
    
    For a security parameter $\lambda$, the KLVY compilation outputs a sequential protocol $\cG^\lambda_{\comp}$ with $k$ rounds.
    Define by $\cG_{\comp} := (\cG^\lambda_{\comp})_{\lambda}$ the \emph{compiled nonlocal game} of $\cG$.
    At the beginning of the protocol, the verifier samples the inputs $\vec{x}$ from the distribution $\mu$.
    In rounds $i \in \{1,\ldots,k-1\}$, the verifier then runs $\Gen(1^\lambda)$ to generate a secret key $\mathsf{sk}_i$ for $\mathsf{QHE}$, uses it to compute and send encryptions $\Enc^{\mathsf{sk}_i}(x_i)$, and receives as responses $\mathrm{a}_i$; in the last round $k$ the verifier sends $x_k$ in the clear and receives $a_k$. The verifier decrypts the responses to obtain $a_i := \Dec^{\mathsf{sk}_i}(\mathrm{a}_i)$ for all $i \in \{1,\ldots,k-1\}$, forming the transcript $(\vec{x},\vec{a})$, which is evaluated using the same predicate $V$ as in $\cG$.
    
    Let $S = (S_{\lambda})_{\lambda}$ be any efficient prover strategy for $\cG_{\comp}$ (i.e., $S_{\lambda}$ implementable by quantum circuits of size $\poly(\lambda)$).
    Denote by $p^{\lambda}(\vec{a}|\vec{x})$ the correlation realized by $S_{\lambda}$ in $\cG_{\comp}^{\lambda}$ at security parameter $\lambda$.
    The \emph{compiled Bell score} of $S$ is
    \begin{align*}
        \gamevalueCompile{\lambda}{\cG_{\comp}, S} :=\;
    \sum_{\vec{x},\vec{a}} \beta_{\vec{x}, \vec{a}} \, p^{\lambda}(\vec{a}|\vec{x}),
    \end{align*}
    where we write the payoff tensor $\beta_{\vec{x}, \vec{a}} := \mu(\vec{x}) V(\vec{x},\vec{a})$ to simplify the notation for the remainder of this manuscript.
\end{definition}

Despite the added difficulty of encryption, the compiled game is designed so that an honest quantum prover can still play optimally. In fact, KLVY~\cite{kalai2023quantum} proved two key properties of their compiler, \emph{classical soundness} and \emph{quantum completeness}, for all $k$-partite games. Classical soundness refers to the fact that any efficient classical prover (one who does not use entanglement or quantum memory) with strategy $S_{\mathrm{classical}}$ cannot win the compiled game with probability higher than the original game's classical value:
\begin{align*}
    \gamevalueCompile{\lambda}{\cG_{\comp}, S_{\mathrm{classical}}} \leq \omega_c(\cG) + \negl(\lambda),
\end{align*}
where $\negl(\lambda)$ is a negligible function that goes to zero faster than the reciprocal of any polynomial in $\lambda$, which is also dependent on the QHE scheme for the compilation.
Quantum completeness refers to the fact that there exists an efficient quantum strategy $S_{\mathrm{complete}}$ for the compiled game whose success probability approaches the original game's quantum value as the security parameter grows: 
\begin{align*}
    \lim_{\lambda\to\infty} \gamevalueCompile{\lambda}{\cG_{\comp}, S_{\mathrm{complete}}} = \omega_q(\cG).
\end{align*}
In particular, if the original nonlocal game exhibits a quantum advantage, then the compiled single-prover protocol also exhibits a classical-quantum gap, at least as big as the nonlocal one.

With the above setup, \emph{quantum soundness} refers to the property that even a malicious quantum prover cannot win with probability exceeding the optimal quantum value of the original $k$-player nonlocal game.

\begin{definition}[Quantum soundness for compiled $k$-partite games]\label{def:QuantumSoundnessCompiledGame}
    Let $\cG$ be a $k$-partite nonlocal game with quantum value $\omega_q(\cG)$ and commuting-operator value $\omega_{\mathrm{qc}}(\cG)$.
    Let $\cG_{\comp} := (\cG^\lambda_{\comp})_{\lambda}$ denote the family of compiled games produced by the KLVY compilation at security parameter $\lambda$.
    
    Let $S=(S_\lambda)_{\lambda}$ be any efficient (QPT) prover strategy for $\cG_{\comp}$, and let $p^\lambda(\vec{a}|\vec{x})$ be the correlation realized by $S_\lambda$ in $\cG^\lambda_{\comp}$.
    We say that the compiler is \emph{quantum sound} (against QPT strategies) if there exists a value
    $B(\cG)\in\{\omega_q(\cG),\,\omega_{\mathrm{qc}}(\cG)\}$ such that for every efficient strategy $S$ there exists a negligible function $\negl_{S}(\lambda)$ with
    \begin{align*}
        \gamevalueCompile{\lambda}{\cG_{\comp}, S} \leq B(\cG) + \negl_{S}(\lambda) \quad \text{for all }\lambda\in\mathds{N}.
    \end{align*}
\end{definition}

Our focus here is precisely on establishing quantitative quantum soundness for \emph{all} compiled \emph{$k$-partite} games: an analytical upper bound on the compiled protocol's quantum winning probability as a function of $\lambda$, showing it stays within a negligible distance of the appropriate benchmark from the original $k$-partite game. 

\subsection{The NPA hierarchy for bounding quantum correlations}\label{sec:PrelimNPAHierarchy}
To upper-bound the winning probabilities of quantum strategies in a $k$-partite nonlocal game $\cG$, a powerful tool is the Navascu\'es-Pironio-Ac\'in (NPA) hierarchy~\cite{navascues2008convergent,pironio2010convergent}. The hierarchy is a sequence of increasingly tight semidefinite programming (SDP) relaxations that characterize the set of quantum-achievable correlations. In broad terms, NPA provides a family of efficiently checkable necessary conditions for a conditional distribution $p(\vec{a}|\vec{x})$ to arise from some quantum strategy. By optimizing the game's payoff over these conditions, one obtains an upper bound on the quantum value. If, at some finite level, a distribution violates an NPA constraint, then it cannot come from any measurement on any shared quantum state. Conversely, as one increases the level (adding higher-degree algebraic constraints), the feasible set converges to the set of commuting-operator quantum correlations, and thus the bounds converge to the commuting-operator value $\omega_{\mathrm{qc}}(\cG)$. In practice, low levels already yield sharp bounds for many games.

To explain the construction in the $k$-partite setting, it is convenient to use an abstract algebraic presentation.

\begin{definition}[Measurement symbols and relations]\label{def:Prelim_NPASymbolRelation}
For each player $i \in [k]$, question $x_i\in X_i$, and answer $a_i \in A_i$, introduce a symbol $f^{(i)}_{a_i|x_i}$ that represents a measurement effect. Without loss of generality the measurements can be assumed to be projective
$${f^{(i)}_{a_i|x_i}}^* = f^{(i)}_{a_i|x_i}, \qquad f^{(i)}_{a_i|x_i}f^{(i)}_{a_i'|x_i} = \delta_{a_i,a_i'}f^{(i)}_{a_i|x_i}, \qquad \sum_{a_i} f_{a_i|x_i}^{(i)} = \id$$
and we impose inter-party commutation $$[f_{a_i|x_i}^{(i)},f^{(j)}_{a_j|x_j}] = 0\qquad (i\neq j).$$ Any commuting-operator strategy is a $*$-representation of these symbols on a Hilbert space with a state $\rho$.
\end{definition}

The NPA hierarchy is phrased in terms of words (monomials) in these symbols and a corresponding moment matrix.

\begin{definition}[Words and level-$n$ moment matrix]
Fix a truncation level $n \in \mathds{N}, n \geq 1$. Let $\cW^n$ be the set of words of length $\leq n$ in the symbols $f_{a|x}^{(i)}$ (and $\id$), modulo the relations above; because different players' symbols commute, words admit a canonical normal order. Given a commuting-operator realization $\left(\rho,\{f_{a|x}^{(i)}\}\right)$  and a purification $\ket{\psi}$ of $\rho$, the level-$n$ moment matrix is the Hermitian matrix $$\Gamma_{w,v}^{(n)} = \bra{\psi}w^* v\ket{\psi}, \qquad \forall w,v\in\cW^n.$$ Then $\Gamma^{(n)}\succeq 0$ (it is a Gram matrix), and it satisfies all linear identities implied by \cref{def:Prelim_NPASymbolRelation}. The degree-$1$ block reads off the correlation: $$p(\vec{a}|\vec{x}) = \Gamma^{(n)}_{\id,\Pi_{i=1}^kf_{a_i|x_i}^{(i)}},$$
whenever $\Pi_{i=1}^kf_{a_i|x_i}^{(i)} \in \cW^{(n)}$. 
\end{definition}
With these ingredients, the level-$n$ NPA relaxation is the SDP that maximizes the game's payoff over all PSD matrices $\Gamma^{(n)}$ obeying the linear identities.

\begin{definition}[Level-$n$ NPA upper bound for a $k$-partite game]\label{def:Prelim_StandardNPA}
Let $\beta_{\vec{a},\vec{x}}$ be the payoff tensor (in predicate form $\beta_{\vec{a},\vec{x}} = \mu(\vec{x})V(\vec{x},\vec{a})$) from the nonlocal game $\cG$. The level-$n$ NPA bound is 
\begin{align*}
    \gamevalueNPA{\cG}{n} := \max_{\Gamma^{n}} {} & \sum_{\vec{a},\vec{x}}\beta_{\vec{a},\vec{x}}\Gamma^{(n)}_{\id,\Pi_{i=1}^kf_{a_i|x_i}^{(i)}}\\
    \text{s.t.}\quad & \Gamma^{(n)}\succeq 0,\\
    &\text{all linear identities induced by the relations on $\cW^{n}$}.
\end{align*}
\end{definition}

Two basic properties encapsulate the usefulness of the hierarchy:
\begin{enumerate}[label=(\alph*)]
    \item \emph{Soundness.} Any commuting-operator strategy produces a feasible $\Gamma^{(n)}$, so for all $n$ $$\omega_{\mathrm{qc}}(\cG) \leq \gamevalueNPA{\cG}{n}.$$
    \item \emph{Monotone convergence.} The sequence $\gamevalueNPA{\cG}{n}$ is nonincreasing and converges to the commuting-operator value: $$\gamevalueNPA{\cG}{1} \geq \gamevalueNPA{\cG}{2} \geq \cdots \searrow \omega_{\mathrm{qc}}(\cG).$$
\end{enumerate}
Thus, optimizing at higher levels tightens the upper bound, and in the limit (including moments of all lengths) one recovers the exact commuting-operator quantum value.

\section{Generalized NPA hierarchy for multipartite quantum sequential scenarios}\label{sec:SeqNPAMain}
In this section, we introduce a generalized NPA hierarchy that converges to multipartite sequential setups as in~\cref{fig:NonlocalCompiledBellGame}(b).
The main novelty of our multipartite generalization is the use of $*$-homomorphisms with appropriate constraints, which allows us to model quantum instruments that are essential in the description of multipartite quantum sequential scenarios.

We begin with the bipartite sequential scenarios in \cref{sec:SeqNPABiparititeCase} to set a foundation for the more complex multipartite scenarios, and to contrast the subnormalized-moment-matrix method of~\cite{klep2025quantitative}.
In \cref{sec:SeqNPATripartiteCase}, we then discuss in detail the construction of the tripartite case for pedagogical purposes and finally introduce the $k$-partite sequential NPA hierarchy in \cref{sec:SeqNPAMultipartiteCase}.

\subsection{Bipartite case revisited}\label{sec:SeqNPABiparititeCase}
As a proof of concept, we first consider the well-understood bipartite case and introduce a sequential generalization of the NPA hierarchy that is more composable than that of~\cite{klep2025quantitative} (see \cref{sec:PrelimSeqNPA} for a quick review).

In a sequential bipartite game $\cG$, Alice receives some state $\sigma$, applies some quantum instrument $\A$, and passes it to Bob for another measurement $\B$.
Using the notation from~\cite[Lemma~12]{baroni2025asymptotic} and thinking in the Heisenberg picture, we first consider generators $\{ f_\by \mid \forall b, y\}$ satisfying the relations $\cR_B$:
\begin{equation*}
    \begin{aligned}
        f_\by^* = f_\by, \quad f_\by f_{b'|y} = \delta_{b,b'} f_\by, \quad \sum_b f_\by = \id.
    \end{aligned}
\end{equation*}
Define Bob's algebra by the universal PVM $C^*$-algebra as
\begin{align*}
    \cA_B = C^*( \{f_\by\}_{b,y} \mid \cR_B).
\end{align*}
In addition, consider generators $\{ f_\abxy \mid \forall a, b, x, y\}$ satisfying the relations $\cR_{AB}$:
\begin{equation*}
    \begin{aligned}
        f_\abxy^* = f_\abxy, \quad f_\abxy f_{a'b'|xy} = \delta_{a,a'} \delta_{b,b'} f_\abxy, \\ 
        \sum_b f_\abxy = \sum_b f_{ab|xy'} \ \forall y, y', \quad \sum_{a,b} f_\abxy = \id.
    \end{aligned}
\end{equation*}
Define the post-Alice-measurement algebra of Bob by
\begin{align*}
    \cA_{AB} = C^*( \{f_\abxy\}_{a,b,x,y} \mid \cR_{AB}),
\end{align*}
denoted as $\cA_{A \to B}$ in~\cite{baroni2025asymptotic}.
For every $a, x$, it is shown in \cite{baroni2025asymptotic} that there exists a (not necessarily unital) *-homomorphism that maps generators to generators
\begin{align*}
    T_\ax: \cA_{B} \to \cA_{AB}, \quad \ f_\by \mapsto f_\abxy,
\end{align*}
which shall be central to our analysis.
A simple --but very effective-- change of perspective with respect to all previous works, is to not model Alice's action as post-measurement states $\phi_{a|x}$, but through CP maps $T_{a|x}$.

For the game $\cG$ and a correlation $p(ab|xy)$, the associated score is $\sum_{a,b,x,y} \beta_{abxy} p(ab|xy)$.
Denote the objective Bell polynomial by
\begin{align*}
    \beta = \sum_{a,b,x,y} \beta_{abxy} f_\abxy \in \sum_{a,x} T_\ax(\cA_B) \subset \cA_{AB}.
\end{align*}

We now construct a variant of the NPA hierarchy with words $f_\abxy$ and leverage the equality $f_\abxy = T_\ax(f_\by)$ to encode the sequential information, such that it converges to the algebraic bipartite sequential strategy.

To this end, define the word sets at level $n$ by
\begin{equation*}
    \begin{aligned}
        \cW^n_{AB}      &:= \bigl\{\,w
         = f_{a_1 b_1|x_1 y_1}\cdots f_{a_k b_k|x_k y_k}\;\big|\;0\le k\le n \bigr\} \subset \cA_{AB},\\[2mm]
        \cW^n_{B}      &:= \bigl\{\,w
         = f_{b_1|y_1}\cdots f_{b_k|y_k}\;\big|\;0\le k\le n \bigr\} \subset \cA_B,\\[2mm]
        T_\ax \bigl(\cW^n_{B}\bigr) &= \bigl\{\,f_{ab_1|xy_1}\cdots f_{ab_k|xy_k}\;\big|\;0\le k\le n \bigr\} \subset T_\ax(\cA_B) \subset \cA_{AB}.\\[2mm]
    \end{aligned}
\end{equation*}
Clearly, $\beta \in \Span(T_\ax \bigl(\cW^n_{B}\bigr)) \subset \Span(\cW_{AB}^n)$ for $n \geq 1$.
Note that for any polynomial $s \in \Span(\cW^n_B)$, the polynomial $T_{\ax}(s) \in \cW^n_{AB}$ has the same degree as $s$, i.e., the homomorphism $T_{\ax}$ does not increase the degree.

Consider the matrix $\Gamma^{(n)}$ indexed by the monomials/words in $\cW^n_{AB}$ with the associated functional $L^{2n}: \cW^{2n}_{AB} \to \mathds{C}$ defined by
\begin{align*}
    L^{2n}( w^* v ) := \Gamma^{(n)}_{w, v}, \quad \forall w, v \in \cW_{AB}^{n}.
\end{align*}
For any $P \in \Span(\cW^{2n}_{AB})$ with $\deg P = k$, denote by 
\begin{align*}
    \Gamma^{(n)}(P)_{w, v} = L^{2n}(w^* P v), \quad \forall w, v \in \cW_{AB}^{n - \lceil k/2 \rceil}
\end{align*}
the localizing matrix of $\Gamma^{(n)}$ at $P$.
This leads to the following NPA-like hierarchy at level $n$:
\begin{equation}\label{eq:Bipartite_SequentialNPASDP}
  \begin{aligned}
    \gamevalueBiSeqNPA{\cG}{n}
    :=\max\; &
        \sum_{a,b,x,y}\beta_{abxy} \, \Gamma^{(n)}_{\id,f_{ab|xy}}   \\[2pt]
    \text{s.t.}\quad
      &\Gamma^{(n)}\succeq 0, \\[2pt]
      &\Gamma^{(n)}_{\id,\id}=1, \\[2pt]
      &\Gamma^{(n)}_{w,v}= \Gamma^{(n)}_{w',v'}
        \quad\bigl(\mathrm{when~} w^{*}v = w'^{*}v'\bigr), \\[6pt]
      &M^{x,x'}_{(w,r),(s,v)}\;:=\;
        \Gamma^{(n)}\!\Bigl(
          \sum_{a} \bigl(T_{a|x}(r^{*}s)-T_{a|x'}(r^{*}s)\bigr)
        \Bigr)_{w,v}\; = 0 \\[2pt]
      &\hphantom{M_{(w,r),(s,v)}:=}
        \substack{\forall x, x', \\
                  \forall w,v\in\cW_{AB}^{n}, \, r,s\in\cW_{B}^{n},\\
                  \deg w+\deg v+\deg r+\deg s\le 2n}
        \quad\text{(operationally‑non‑signaling)}.
  \end{aligned}
\end{equation}
The last operationally-non-signaling constraint is equivalent to
\begin{align*}
    \sum_a L^{2n}(P^* T_{\ax}(S) Q) - \sum_a L^{2n}(P^* T_{a|x'}(S) Q) = 0
\end{align*}
for all $x, x'$, $P, Q \in \Span(W^n_{AB})$, and $S \in \Span(W^n_{B})$ such that $\deg P + \deg Q + \deg S \leq 2n$.
A moment matrix $\Gamma^{(n)}$ is said to be a feasible solution of \cref{eq:Bipartite_SequentialNPASDP} if it satisfies all the constraints but does not necessarily maximize the score, and is said to be an optimal feasible solution of \cref{eq:Bipartite_SequentialNPASDP} if it satisfies all constraints and maximizes the score.

The physical submatrices of $\Gamma^{(n)}$ include the block indexed by $T_\ax(\cW^n_{B}) \times T_\ax(\cW^n_{B})$ corresponding to Bob's measurement, subject to the output-input $(a,x)$ by Alice's measurement.
Thus, the hierarchy \cref{eq:Bipartite_SequentialNPASDP} is stricter than the bipartite sequential NPA hierarchy (\cref{eq:Prelim_BipartiteSequentialNPASDP}) of~\cite{klep2025quantitative} and the nice SOS hierarchy~\cite[Eqs.~4.2 and~4.4]{cui2025convergent} by identifying this submatrix with its Riesz functional $\sigma_{a|x}$.
While not necessary in the bipartite scenarios, the inclusion of the ``nonphysical'' off-diagonal terms is critical in the convergent proofs beyond three parties in the following subsections.

We present the following convergence theorem on the hierarchy \cref{eq:Bipartite_SequentialNPASDP}.
\begin{theorem}\label{thm:Bipartite_Completeness}
    Let $p(ab|xy)$ be a correlation for the nonlocal game $\cG$.
    The following statements are equivalent:
    \begin{enumerate}[label=(\roman*)]
        \item The correlation $p(ab|xy)$ arises from a bipartite sequential operationally-non-signaling strategy.
        \item The correlation $p(ab|xy)$ arises from a bipartite commuting operator strategy.
        \item There exists a family of $\{\Gamma^{(n)}\}_n$ of feasible solutions to \cref{eq:Bipartite_SequentialNPASDP} such that $p(ab|xy) = \Gamma^{(n)}_{\id,f_{ab|xy}}$ for all $n$.
    \end{enumerate}
    Consequently, $\gamevalueBiSeqNPA{\cG}{n} \searrow \gamevalueqcopt{\cG}$ monotonically as $n \to \infty$.
\end{theorem}
\begin{proof}
    It is clear that both \textit{(i)} and \textit{(ii)} imply \textit{(iii)}.
    The equivalence between $\textit{(i)}$ and \textit{(ii)} is well-established in, e.g.,~\cite{kulpe2024bound}.

    The statement \textit{(iii)} implies \textit{(i)} is a consequence of the paragraph preceding the theorem, combined with the convergence property of the sequential NPA hierarchy of~\cite{klep2025quantitative}.
    Nonetheless, we do another proof that connects to~\cite{baroni2025asymptotic} better.
    
    To this end, one has $\Span(\cW^n_{AB}) \to \cA_{AB}$ as $n \to \infty$.
    Since all generators $f_\abxy$ are projective, each entry of $\Gamma^{(n)}$ is necessarily bounded for all $n$.
    Then the standard Banach-Alaoglu argument implies the existence of a convergent subsequence $\Gamma^{(n_k)}$ and an infinite moment matrix $\Gamma$, such that
    \begin{align*}
        \Gamma^{(n_k)}_{w,v} \to \Gamma_{w, v}
    \end{align*}
    for all $w, v$ and all $n_k \geq \max\{\deg(w), \deg(v)\}$.
    That is, we have recovered a state $\sigma$ on $\cA_{AB}$ via
    \begin{align*}
        \sigma: \cA_{AB} &\to \mathds{C},\\
        w^*v &\mapsto \Gamma_{w, v},
    \end{align*}
    such that
    \begin{align*}
        p(ab|xy) = \sigma ( T_\ax (f_\by) )
    \end{align*}
    and 
    \begin{align*}
        \sum_a \sigma(P^* T_{\ax}(S) Q) - \sigma(P^* T_{a|x'}(S) Q) = 0
    \end{align*}
    for all $S \in \cA_B$ and $P, Q \in \cA_{AB}$.
    We are done by identifying $\sigma$ with an asymptotically-secured $C^*$-algebraic compiled strategy for two players~\cite[Definition~17 and Theorem~14]{baroni2025asymptotic}.

    For completeness, let us give an explicit proof.
    Consider the GNS representation $(\cH, \pi, \ket{\Omega})$ of $\sigma$, where $\cH$ is a Hilbert space, $\pi: \cA_{AB} \to B(\cH)$ is a $*$-representation, and $\ket{\Omega} \in \cH$ such that $\sigma(P) = \sandwich{\Omega}{\pi(P)}{\Omega}$ for all $P \in \cA_{AB}$.
    It follows that, for any $x, x'$, $S \in \cA_B$, and $P, Q \in \cA_{AB}$,
    \begin{align*}
        0 = \sandwich{\Omega}{\pi(P)^* \pi( \sum_a T_{\ax}(S) - \sum_a T_{a|x'}(S) ) \pi(Q)}{\Omega}.
    \end{align*}
    By cyclicity, this implies that
    \begin{align*}
        \sum_a \pi \circ T_{\ax} = \sum_a \pi \circ T_{a|x'} := \bar{T}: \cA_B \to B(\cH),
    \end{align*}
    where $\pi \circ T_{\ax}$ is completely positive map $\cA_B \to B(\cH)$ (due to *-homomorphisms being completely positive) dominated by the completely positive map $\bar{T}$.
    Thus, $(\ket{\Omega}, \pi \circ T_{\ax})$ defines a $C^*$-algebraic sequential correlation as in~\cite[Definition~22]{baroni2025asymptotic}, finishing \textit{(iii)} $\implies$ \textit{(i)}.

    Let us show \textit{(iii)} $\implies$ \textit{(ii)}, analogously to~\cite[Theorem~16]{baroni2025asymptotic} for further demonstration.
    Let $(\cK, \pi_{\cK}, V)$ be the minimal Stinespring dilation of $\bar{T}$ such that, for all $S \in \cA_B$,
    \begin{align*}
        \bar{T}(S) = \sum_a \pi \circ T_{\ax}(S) = V^* \pi_\cK(S) V.
    \end{align*}
    Arveson's Radon-Nikodym derivative~\cite[Theorem~1.4.2]{arveson1969subalgebras}, see also~\cite{raginsky2003radon}, then implies the existence of unique positive operator $F_{a|x} \in \pi_\cK(\cA_B)'$ such that
    \begin{align*}
        \pi \circ T_\ax(S) = V^* F_\ax \pi_\cK(S) V.
    \end{align*}
    The uniqueness of the minimal dilation with $\sum_a T_{\ax} = \bar{T}$ further imposes $\sum_a F_{\ax} = \id$, i.e., $F_{\ax}$ form POVMs for every $x$.
    We are done by identifying the state $\sigma$ with $V\ket{\Omega}$.
\end{proof}

We now analyze some properties of the novel sequential hierarchy that will be useful later; more precisely, we discuss the stopping criterion and strict feasibility.

\begin{definition}\label{def:Biparitite_FlatnessCondition}
    For $n \in \mathds{N}$, let $\Gamma^{(n)}$ be the solution for \cref{eq:Bipartite_SequentialNPASDP} at level $n$ for some nonlocal game $\cG$.
    Consider its block form
    \begin{align*}
        \Gamma^{(n)} = \begin{pmatrix}
            \Gamma^{(n-1)} & M \\
            M^* & N
        \end{pmatrix},
    \end{align*}
    where $\Gamma^{(n-1)}$ is the principal block indexed by words in $\cW^{n-1}_{AB}$.
    We say that the solution $\Gamma^{(n)}$ is flat (or has a rank-loop) if
    \begin{align*}
        \mathrm{rank}(\Gamma^{(n)}) = \mathrm{rank}(\Gamma^{(n-1)}) < \infty.
    \end{align*}
\end{definition}

\begin{proposition}\label{prop:Biparitite_FlatnessCondition}
    If the hierarchy of \cref{eq:Bipartite_SequentialNPASDP} for $\cG$ admits a flat optimal solution at some finite level $n$, then the hierarchy attains its limiting value at level $n$, and the flat solution yields a finite-dimensional optimal tensor-product quantum strategy.
    In particular, $\gamevalueBiSeqNPA{\cG}{n} = \omega_{\mathrm{qc}}(\cG) = \omega_{\mathrm q}(\cG)$.
\end{proposition}
\begin{proof}
    This follows from the standard flat-extension/GNS argument, see e.g.~\cite[Theorem~10]{navascues2008convergent}. Flatness gives a finite-dimensional representation realizing the level-$n$ optimal value.
    Since this representation is a genuine quantum strategy, the level-$n$ value is bounded above by $\omega_{\mathrm{qc}}(\cG)$, while the hierarchy is an outer relaxation and hence is bounded below by $\omega_{\mathrm{qc}}(\cG)$.
    Therefore equality holds, and finite-dimensional Tsirelson-type equivalence~\cite{scholz2008tsirelson,xu2025quantitative} gives the tensor-product strategy and $\omega_{\mathrm{qc}}(\cG)=\omega_{\mathrm q}(\cG)$.
    
    Note that the equality $\mathrm{rank}(\Gamma^{(n)}) = \mathrm{rank}(\Gamma^{(n-1)})$ is sufficient since this already enforces the operationally-non-signaling condition on the generator $f_\abxy$ in the resulting finite-dimensional representation and thereby can propagate to higher degree words.
\end{proof}
We do not use or claim the converse.
A finite-dimensional optimal strategy gives a genuine finite-rank moment sequence, but this alone does not rule out higher-valued non-flat pseudo-moment solutions at finite levels.
Note that the existence of flat optimal solutions does not guarantee that numerical algorithms will find it in practice.
In fact, it is possible that there exist infinitely many inequivalent finite-dimensional optimal strategies, leading the SDP solver to return any convex mixture of them.

\begin{proposition}\label{prop:Bipartite_StrictFeasible}
    For every level $n$, the SDP in \cref{eq:Bipartite_SequentialNPASDP} is strictly feasible.
    That is, there exists a moment matrix $\Gamma^{(n)}_{\strict} \succ 0$ that satisfies all the linear constraints of \cref{eq:Bipartite_SequentialNPASDP}.
\end{proposition}
\begin{proof}
    The unconstrained NPA hierarchy with commuting PVMs $\{\A\}, \{\B\}$ admits a strictly feasible moment matrix at every level.
    This follows from the faithfulness of the left-regular GNS representation of the universal $*$-algebra generated by $\A$ and $\B$.
    Thus, for each $n$ there exists a full-rank moment matrix $\Gamma^{(n)}_\NPA \succ 0$ for the standard NPA hierarchy.
    (See~\cite[Appendix~C]{tavakoli2024semidefinite} for a more explicit argument.)

    Define $f_\abxy := \A \B$.
    Then every word in the sequential hierarchy of \cref{eq:Bipartite_SequentialNPASDP} is also a word in the larger algebra generated by $\A$, $\B$.
    Therefore, we let $\Gamma^{(n)}_{\strict}$ be the principle submatrix of $\Gamma^{(n)}_\NPA$ corresponding to all the $f_\abxy$-words, which can be straightforwardly checked satisfy all constraints of \cref{eq:Bipartite_SequentialNPASDP}.
    Finally, since a principal submatrix of a positive definite matrix is itself positive definite, $\Gamma^{(n)}_{\strict} \succ 0$.
\end{proof}

\subsection{Tripartite case}\label{sec:SeqNPATripartiteCase}
Now, we showcase the construction of the tripartite sequential NPA hierarchy in detail to inspire the general multipartite case.
As in \cref{fig:NonlocalCompiledBellGame}(b), the game begins with Alice $A$ receiving and answering with the pair $(x,a)$, followed by Bob $B$ with the pair $(y,b)$, and ends with Charlie $C$ with the pair $(z,c)$, during which the operationally non-signaling condition is respected by Alice and Bob.

To describe this scenario, consider generators $\{f_\cz \mid \forall c, z\}$ satisfying the relation $\cR_C$:
\begin{equation*}
    \begin{aligned}
        f_\cz^* = f_\cz, \quad f_\cz f_{c'|z} = \delta_{c,c'} f_\cz, \quad \sum_c f_\cz = \id,
    \end{aligned}
\end{equation*}
and define Charlie's algebra by the universal PVM $C^*$-algebra
\begin{align*}
    \cA_C = C^*( \{f_\cz\}_{c,z} \mid \cR_C).
\end{align*}
Next, consider generators $\{f_\bcyz \mid \forall b, c, y, z\}$ satisfying the relation $\cR_{BC}$:
\begin{equation}
    \begin{aligned}
        f_\bcyz^* = f_\bcyz, \quad f_\bcyz f_{b'c'|yz} = \delta_{b,b'} \delta_{c,c'} f_\bcyz, \\ 
        \sum_c f_\bcyz = \sum_c f_{bc|yz'} \ \forall z, z', \quad \sum_{b,c} f_\bcyz = \id.
    \end{aligned}
\end{equation}
Define post-Bob-measurement algebra of Charlie by
\begin{align*}
    \cA_{BC} = C^*( \{f_\bcyz\}_{b,c,y,z} \mid \cR_{BC}).
\end{align*}
denoted as $\cA_{B \to C}$~\cite{baroni2025asymptotic}.
Then, consider the generators $\{f_\abcxyz \mid \forall a, b, c, x, y, z \}$ satisfying the relation $\cR_{ABC}$:
\begin{equation*}
    \begin{aligned}
        f_\abcxyz^* = f_\abcxyz, \quad
        f_\abcxyz f_{a'b'c'|xyz}= \delta_{a,a'} \delta_{b,b'} \delta_{c,c'} f_\abcxyz, \\
        \sum_c f_\abcxyz = \sum_c f_{abc|xyz'}, \quad \sum_{b,c} f_\abcxyz = f_{abc|xy'z'}, \quad \sum_{a,b,c} f_\abcxyz = \id,
    \end{aligned}
\end{equation*}
and define the post-Alice-Bob-measurement algebra of Charlie by
\begin{align*}
    \cA_{ABC} = C^*( \{f_\abcxyz\}_{a,b,c,x,y,z} \mid \cR_{ABC}),
\end{align*}
denoted as $\cA_{A \to B \to C}$ in~\cite{baroni2025asymptotic}.
Finally, define the natural *-homomorphisms for $a, b, x, y$ by 
\begin{equation*}
    \begin{aligned}
        T_\by &: \cA_C \to \cA_{BC}, \quad f_{c|z} \mapsto f_\bcyz, \\
        T_\ax &: \cA_{BC} \to \cA_{ABC}, \quad f_\bcyz \mapsto f_\abcxyz.
    \end{aligned}
\end{equation*}

Similarly, for a tripartite nonlocal game $\cG$ and a correlation $p(abc|xyz)$, the associated score
\begin{align*}
    \sum_{a,b,c,x,y,z} \beta_{abcxyz} p(abc|xyz)
\end{align*}
results in the objective polynomial
\begin{align*}
    \beta = \sum_{a,b,c,x,y,z} \beta_{abcxyz} f_\abcxyz \in \sum_{a,b,x,y} T_\ax T_\by(\cA_{C}) \subset \cA_{ABC}.
\end{align*}

Analogous to \cref{sec:SeqNPABiparititeCase}, define the word sets at level $n$:
\begin{equation*}
    \begin{aligned}
        \cW^n_{ABC} &:= \Bigl\{ f_{a_1 b_1 c_1|x_1 y_1 z_1}\cdots f_{a_k b_k c_k|x_k y_k z_k}\;|\;0\!\le\!k\!\le\!n \Bigr\} \subset \cA_{ABC}, \\
        \cW^n_{BC} &:= \Bigl\{ f_{b_1 c_1|y_1 z_1}\cdots f_{b_k c_k|y_k z_k}\;|\;0\!\le\!k\!\le\!n \Bigr\} \subset \cA_{BC}, \\
        \cW^{n}_{C} &:= \Bigl\{ f_{c_1|z_1}\cdots f_{c_k|z_k}\;|\;0\!\le\!k\!\le\!n \Bigr\} \subset \cA_C,
    \end{aligned}
\end{equation*}
along with $T_{\by}\!\bigl(\cW^{\,n}_{C}\bigr) \subset \cA_{BC}$ and $T_{\ax}T_{\by}\!\bigl(\cW^{\,n}_{C}\bigr) \subset \cA_{ABC}$.
Again, $\beta \in \Span(\cW_{ABC}^n)$ for $n \geq 1$.

Consider the matrix $\Gamma^{(n)}$ indexed by monomials/words in $\cW^n_{ABC}$.
Using the same notation as in \cref{sec:SeqNPABiparititeCase}, we define the following tripartite sequential NPA-like hierarchy at level $n$:
\begin{equation}\label{eq:Tripartite_SequentialNPASDP}
  \begin{aligned}
    \gamevalueTriSeqNPA{\cG}{n}
    \;:=\;
    \max\quad
      & \sum_{a,b,c,x,y,z} \beta_{abcxyz}\;
        \Gamma^{(n)}_{\id, f_{abc|xyz}}\\
    \text{s.t.}\quad
      &\Gamma^{(n)} \succeq 0,\\
      &\Gamma^{(n)}_{\id,\id}=1,\\
      &\Gamma^{(n)}_{w,w'}=\Gamma^{(n)}_{v,v'}
        \quad\text{whenever }w^{*}w'=v^{*}v',\\[2pt]
      &M^{x,x'}_{(w,r),(s,v)}\;:=\;
        \Gamma^{(n)}\!\Bigl(
          \bigl(\sum_{a}T_{a|x}(r^{*}s)-T_{a|x'}(r^{*}s)\bigr)
        \Bigr)_{w,v}\;= 0 \\[2pt]
      &\hphantom{M_{(w,r),(s,v)}:=}
        \substack{\forall x, x', \, \forall w,v\in\cW_{ABC}^{n},\\
                  r,s\in\cW_{BC}^{n},\\
                  \deg w+\deg v+\deg r+\deg s\le 2n}
        \quad\text{(Alice operationally‑non‑signaling)} \\[2pt]
        &N^{y,y'}_{(r,t),(s,u)}\;:=\;
        \Gamma^{(n)}\!\Bigl(
          T_\ax\bigl(\sum_{b}T_{b|y}(t^{*}u)-T_{b|y'}(t^{*}u)\bigr)
        \Bigr)_{T_\ax(r),T_\ax(s)}\;= 0 \\[2pt]
      &\hphantom{N_{(r,t),(s,u)}:=}
        \substack{\forall\, a,x, y, y', 
        \\ r,s\in\cW_{BC}^{n},\\
                  t,u\in\cW_{C}^{n},\\
                  \deg r+\deg s+\deg t+\deg u\le 2n}
        \quad\text{(Bob operationally‑non‑signaling)}
  \end{aligned}
\end{equation}
Here, let $L^{2n}$ be the normalized positive linear map associated with $\Gamma^{(n)}$, then Alice operationally-non-signaling constraint is equivalent to
\begin{align*}
    \sum_a L^{2n}(P^* T_{\ax}(S) Q) - \sum_a L^{2n}(P^* T_{a|x'}(S) Q) = 0
\end{align*}
for all $x, x'$, $P, Q \in \Span(\cW^n_{ABC})$, and $S \in \Span(\cW^n_{BC})$ such that $\deg P + \deg Q + \deg S \leq 2n$.
In addition, Bob operationally-non-signaling constraint is equivalent to
\begin{align*}
    \sum_b L^{2n} \circ T_\ax (R^* T_{\by}(O) S) - \sum_b L^{2n} \circ T_\ax(R^* T_{b|y'}(O) S) = 0
\end{align*}
for all $a, x, y, y'$, $R, S \in \Span(\cW^n_{BC})$, and $O \in \Span(\cW^n_{C})$ such that $\deg R + \deg S + \deg O \leq 2n$.

A moment matrix is said to be a feasible solution of \cref{eq:Tripartite_SequentialNPASDP} if it satisfies all the constraints but does not necessarily maximize the score, and is said to be an optimal feasible solution of \cref{eq:Tripartite_SequentialNPASDP} if it satisfies all constraints and maximizes the score.

Similarly to the bipartite case, $\cW^n_{C} \times \cW^n_{C}$ plays the role of Charlie's measurement, the block of $\cW^{\,n}_{\by} \times \cW^{\,n}_{\by}$ corresponds to the case when Bob measures afterwards with $(b,y)$, and the block $\cW^{\,n}_{\abxy} \times \cW^{\,n}_{\abxy}$ represents the case of Alice measuring with $(a,x)$ after Bob and Charlie's measurement.
We analogously have the following theorem.
\begin{theorem}\label{thm:Tripartite_Completeness}
    Let $p(abc|xyz)$ be a correlation for the tripartite nonlocal game $\cG$.
    The following statements are equivalent:
    \begin{enumerate}[label=(\roman*)]
        \item The correlation $p(abc|xyz)$ arises from a tripartite sequential operationally-non-signaling strategy.
        \item The correlation $p(abc|xyz)$ arises from a tripartite commuting operator strategy.
        \item There exists a family of $\{\Gamma^{(n)}\}_n$ of feasible solutions to \cref{eq:Tripartite_SequentialNPASDP} such that $p(abc|xyz) = \Gamma^{(n)}_{\id, f_{abc|xyz}}$ for all $n$.
    \end{enumerate}
    Consequently, $\gamevalueTriSeqNPA{\cG}{n} \searrow \gamevalueqcopt{\cG}$ monotonically as $n \to \infty$.
\end{theorem}
\begin{proof}
    It is straightforward to check that both \textit{(i)} and \textit{(ii)} imply \textit{(iii)}.
    The equivalence between \textit{(i)} and \textit{(ii)} is shown by~\cite{baroni2025asymptotic} thanks to the new chain rule for Arveson's Radon-Nikodym derivatives.

    The direction \textit{(iii)} $\implies$ \textit{(i)} is almost the same as the proof of \cref{thm:Bipartite_Completeness}.
    Again, $\Span(\cW^n_{ABC}) \to \cA_{ABC}$ as $n \to \infty$ with projective generators.
    Then by the Banach-Alaoglu Theorem there exists a weak-* convergent subsequence $\Gamma^{(n_k)} \to \Gamma$.
    That is, we have obtained a state on $\cA_{ABC}$ via
    \begin{align*}
        \sigma:  \cA_{ABC} &\to \mathds{C},\\
        w^*v &\mapsto \Gamma_{w, v}
    \end{align*}
    such that
    \begin{align*}
        p(abc|xyz) = \sigma(f_\abcxyz) = \sigma( T_\ax T_\by(f_\cz)).
    \end{align*}
    Furthermore, we have the operational-non-signaling for Alice
    \begin{align*}
        \sum_a \sigma(P^* T_{\ax}(S) Q) - \sum_a \sigma(P^* T_{a|x'}(S) Q) = 0
    \end{align*}
    for all $P, Q \in \cA_{ABC}$ and $S \in \cA_{BC}$; and for Bob
    \begin{align*}
        \sum_b \sigma \circ T_\ax (R^* T_{\by}(O) S) - \sum_b \sigma \circ T_\ax(R^* T_{b|y'}(O) S) = 0
    \end{align*}
    for all $R, S \in \cA_{BC}$ and $O \in \cA^n_{C}$.

    We are done by identifying $\sigma$ with a asymptotically-secured $C^*$-algebraic compiled strategy for three players as in~\cite[Definition~17 and Theorem~14]{baroni2025asymptotic} and then invoking~\cite[Theorem~16 or~17]{baroni2025asymptotic}.
\end{proof}

We omit the flatness condition and the strict feasibility for the tripartite hierarchy since it is straightforward, and instead present the $k$-partite variant directly in the next subsection.

\subsection{General multipartite case}\label{sec:SeqNPAMultipartiteCase}
We are ready to tackle the general $k$-partite sequential quantum scenarios for any $k \geq 2$.

First, let us develop a notation for any $k$-partite sequential algebras and the corresponding NPA hierarchy.
In our sequential convention, we begin with the $k$th party and then $(k-1)$-th until the $1$-st party, i.e., $k \to k-1 \to \cdots \to 1$.
Write $[j] := \{1, 2, \dots, j\}$ and denote $a_1 \cdots a_j, a'_1 \cdots a'_j,  x_1 \cdots x_j$ by $a_{[j]}, a'_{[j]}, x_{[j]}$, respectively, when there is no ambiguity.

For any $j \in [k]$, consider the generators/letters $\{f_{a_{[j]} | x_{[j]}} \mid \forall a_{[j]}, x_{[j]}\}$ satisfying the relation $\cR_{[j]}$:
\begin{equation*}
    \begin{aligned}
        f_{a_{[j]} | x_{[j]}}^* = f_{a_{[j]} | x_{[j]}}, \quad f_{a_{[j]} | x_{[j]}} f_{a'_{[j]} | x_{[j]}} = \delta_{a_{[j]}, a'_{[j]}} f_{a_{[j]} | x_{[j]}}, \quad \sum_{a_{[j]}} f_{a_{[j]} | x_{[j]}} = \id, \\
        \sum_{a_1, \dots a_l} f_{a_1 \cdots a_l a_{l+1} \cdots a_j | x_1 \cdots x_l x_{l+1} \cdots x_j} = \sum_{a_1, \dots a_l} f_{a_1 \cdots a_l a_{l+1} \cdots a_j | x'_1 \cdots x'_l x_{l+1} \cdots x_j}, \, \forall l \in [j].
    \end{aligned}
\end{equation*}
The corresponding universal $C^*$-algebra is defined as
\begin{align}
    A_{[j]} = C^*( \{ f_{a_{[j]} | x_{[j]}} \}_{a_{[j]}, x_{[j]}} \mid \cR_{[j]}).
\end{align}
For every $j$ such that $1 < j \leq k$, there exists a natural *-homomorphism for $a_k, x_k$ such that
\begin{equation*}
    \begin{aligned}
        T^j_{a_j | x_j}: \cA_{[j-1]} \to \cA_{[j]}, \quad f_{a_1 \cdots a_{j-1} | x_1 \cdots x_{j-1}} \mapsto f_{a_1 \cdots a_{j-1} a_j | x_1 \cdots x_{j-1} x_j}.
    \end{aligned}
\end{equation*}
That is, in the Heisenberg picture, the sequentiality is captured as
\begin{align*}
    \cA_{[1]} \xrightarrow{T^2_{a_2 | x_2}} \cA_{[2]} \xrightarrow{T^3_{a_3 | x_3}} \cdots \xrightarrow{T^{k-1}_{a_{k-1} | x_{k-1}}} \cA_{[k-1]} \xrightarrow{T^k_{a_k | x_k}} \cA_{[k]}
\end{align*}
with $\cA_{[1]} = C^*( \{f_{a_1|x_1}\}_{a_1, x_1} \mid \cR_{[1]})$ as the end party (e.g., Bob in the bipartite case and Charlie in the tripartite case).

For a $k$-partite nonlocal game $\cG$ with the correlation $p(\akxk)$, the associated score objective Bell polynomial is
\begin{align*}
    \beta = \sum_{a_{[k]}, x_{[k]}} \beta_{a_{[k]} x_{[k]}} f_\akxk \in \sum_{a_{[k]}, x_{[k]}} T^k_{a_k | x_k} \cdots T^2_{a_2 | x_2} ( \cA_{[1]}) \subset \cA_{[k]}.
\end{align*}

Let us define the $n$-th level of the $k$-partite sequential NPA hierarchy for game $\cG$.
We consider the word sets at level $n$ for each $j$:
\begin{align}
    \cW^n_{[j]} := \{ w \in \cA_{[k]} \mid  0 \leq \deg(w) \leq n \} \subset \cA_{[j]},
\end{align}
with $\beta \in \Span(\cW^n_{[k]})$ and $T^k_{a_k | x_k} \cdots T^2_{a_2 | x_2} ( \cW^n_{[1]}) \subset \cW^n_{[k]}$.
For the matrix $\Gamma^{(n)}$ indexed by $\cW^n_{[k]}$, the $k$-partite sequential NPA hierarchy at level is defined as:
\begin{equation}\label{eq:kpartite_SequentialNPASDP}
  \begin{aligned}
    \gamevaluekSeqNPA{\cG}{n}
    \;:=\;
    \max \quad
      & \sum_{a_{[k]}, x_{[k]}} \beta_{a_{[k]} x_{[k]}}\;
        \Gamma^{(n)}_{\id, f_\akxk}\\
    \text{s.t.}\quad
      &\Gamma^{(n)} \succeq 0,\\
      &\Gamma^{(n)}_{\id,\id}=1,\\
      &\Gamma^{(n)}_{w,w'}=\Gamma^{(n)}_{v,v'}
        \quad\text{whenever }w^{*}w'=v^{*}v',\\[2pt]
      &\sum_{a_j} \Gamma^{(n)}\!\Bigl(
          T^k_{a_k|x_k} \cdots T^{j+1}_{a_{j+1} | x_{j+1}} \bigl( T^j_{a_j | x_j}(r^{*}s)-T^j_{a_j | x'_j}(r^{*}s)\bigr)
        \Bigr)_{w,v}\;= 0 \\[2pt]
      &\hphantom{M_{(w,r),(s,v)}:=}
        \substack{\forall j \in [2,k] \cap \mathds{N}, \\
            \forall a_{j+1}, \dots, a_{k}, x'_j, x_j,\dots, x_{k}, \, \forall \, r,s \in \cW_{[j-1]}^{n},\\
            \forall\, w,v \in T^k_{a_k|x_k} \cdots T^{j+1}_{a_{j+1} | x_{j+1}}(\cW_{[j]}^{n}),\\
                  \deg w+\deg v+\deg r+\deg s\le 2n}
        \quad\text{(party-$j$ operationally‑non‑signaling)}.
  \end{aligned}
\end{equation}
Here, let $L^{2n}$ be the normalized positive linear map associated with $\Gamma^{(n)}$, then party-$j$ operationally-non-signaling constraint is equivalent to
\begin{align*}
    \sum_{a_j} L^{2n} \circ T^k_{a_k|x_k} \cdots T^{j+1}_{a_{j+1} | x_{j+1}} \bigr( P^* T^j_{a_j | x_j}(S) Q\bigr) -  \sum_{a_j} L^{2n} \circ T^k_{a_k|x_k} \cdots T^{j+1}_{a_{j+1} | x_{j+1}} \bigr( P^* T^j_{a_j | x'_j}(S) Q\bigr) = 0
\end{align*}
for all $x'_j, x_j, \dots, x_k, a_{j+1}, \dots a_k$, $P, Q \in \Span(\cW^n_{[j]})$, and $S \in \Span(\cW^n_{[j-1]})$.
This corresponds to the degree $n$ relaxation of~\cite[Eq.~(38)]{baroni2024quantum}.

With the notation introduced, we state all the $k$-partite results analogous to the bipartite and tripartite cases.
\begin{theorem}\label{thm:kpartite_Completeness}
    Let $p(a_{[k]}|x_{[k]})$ be a correlation for the $k$-partite nonlocal game $\cG$.
    The following statements are equivalent:
    \begin{enumerate}[label=(\roman*)]
        \item The correlation $p(a_{[k]}|x_{[k]})$ arises from a $k$-partite sequential operationally-non-signaling strategy.
        \item The correlation $p(a_{[k]}|x_{[k]})$ arises from a $k$-partite commuting operator strategy.
        \item There exists a family of $\{\Gamma^{(n)}\}_n$ of feasible solutions to \cref{eq:kpartite_SequentialNPASDP} such that $p(a_{[k]}|x_{[k]}) = \Gamma^{(n)}_{\id, f_\akxk}$ for all $n$.
    \end{enumerate}
    Consequently, $\gamevaluekSeqNPA{\cG}{n} \searrow \gamevalueqcopt{\cG}$ monotonically as $n \to \infty$.
\end{theorem}
\begin{proof}
    This follows from a routine inductive generalization of the proof for \cref{thm:Bipartite_Completeness,thm:Tripartite_Completeness}.
\end{proof}

We now present the flatness condition and strict feasibility of the $k$-partite hierarchy achieved by a direct inductive generalization.
\begin{definition}\label{def:kpartite_FlatnessCondition}
    For $n \in \mathds{N}$, let $\Gamma^{(n)}$ be the solution for \cref{eq:kpartite_SequentialNPASDP} at level $n$ for some nonlocal game $\cG$.
    Consider its block form
    \begin{align*}
        \Gamma^{(n)} = \begin{pmatrix}
            \Gamma^{(n-1)} & M \\
            M^* & N
        \end{pmatrix},
    \end{align*}
    where $\Gamma^{(n-1)}$ is the principal block indexed by words in $\cW^{n-1}_{[k]}$.
    We say that the solution $\Gamma^{(n)}$ is flat (or has a rank-loop) if
    \begin{align*}
        \mathrm{rank}(\Gamma^{(n)}) = \mathrm{rank}(\Gamma^{(n-1)}) < \infty.
    \end{align*}
\end{definition}

\begin{proposition}\label{prop:kpartite_FlatnessCondition}
    If the hierarchy of \cref{eq:kpartite_SequentialNPASDP} for $\cG$ admits a flat optimal solution at some finite level $n$, then the hierarchy attains its limiting value at level $n$, and the flat solution yields a finite-dimensional optimal tensor-product quantum strategy.
    In particular, $\gamevaluekSeqNPA{\cG}{n} = \omega_{\mathrm{qc}}(\cG) = \omega_{\mathrm q}(\cG)$.
\end{proposition}
\begin{proof}
    This can be proven analogously to \cref{prop:Biparitite_FlatnessCondition}.
\end{proof}

\begin{proposition}\label{prop:kpartite_StrictFeasible}
    For every level $n$, the SDP in \cref{eq:kpartite_SequentialNPASDP} is strictly feasible.
    That is, there exists a moment matrix $\Gamma^{(n)}_{\strict} \succ 0$ that satisfies all the linear constraints of \cref{eq:kpartite_SequentialNPASDP}.
\end{proposition}
\begin{proof}
    This is inductive from the proof of \cref{prop:Bipartite_StrictFeasible}.
\end{proof}

\section{Quantitative quantum soundness for multipartite compiled nonlocal games}\label{sec:SoundnessCompileMain}
In this section, we give our main hierarchy-based quantitative upper bound for multipartite compiled nonlocal games (as in \cref{fig:NonlocalCompiledBellGame}(b)), from which the quantum-value soundness statements follow under finite-level convergence or flat optimality.
We begin with introducing the necessary notations in \cref{sec:SoundnessCompileNotation}.
Then in \cref{sec:GeometricProof}, since the old technique of signaling decomposition of~\cite{klep2025quantitative} cannot be extended to multipartite scenarios, we give novel proofs for multipartite signaling decomposition to show that the compiled strategy at any $\lambda$ is negligibly different from some solutions of \cref{eq:kpartite_SequentialNPASDP} for any NPA level $n$.
This new approach uses geometric arguments involved with the projection on the convex set of constraints of \cref{eq:kpartite_SequentialNPASDP}, and then fixing the positivity with strict feasibility followed by a renormalization, see \cref{fig:GeometryOfProof} for illustration.
Finally in \cref{sec:PiecesTogetherCorollary}, quantitative quantum soundness is prove by relating the above closeness result and the convergence of the sequential NPA hierarchy together.
We refer to \cref{fig:big-scheme} for a more detailed graphical representation of the proofs of our manuscript.

\subsection{Notations preliminaries}\label{sec:SoundnessCompileNotation}
According to~\cite{baroni2025asymptotic}, at the security parameter $\lambda$, an efficient (quantum polynomial-time) strategy $S = (S_{\lambda})$ of a $k$-partite compiled nonlocal game gives rise to states $\sigma^\lambda: \cA_{[k]} \to \mathds{C}$ such that
\begin{align*}
    p^{\lambda}(a_{[k]}|x_{[k]}) = \sigma^\lambda( T^k_{a_k | x_k} \cdots T^2_{a_2 | x_2} (f_{a_{[1]} | x_{[1]}}) ).
\end{align*}
Instead of the operationally-non-signaling condition as in \cref{eq:kpartite_SequentialNPASDP}, we only have a weakly-non-signaling condition.
That is, for all $j \in [2, k] \cap \mathds{N}$ and for all $x'_j, x_j, \dots, x_k, a_{j+1}, \dots a_k$ and for every operator $P, Q \in \cA_{[j]}$ and $R \in \cA_{[j-1]}$, there exists a negligible function $\negl_{S,P,R,Q}(\lambda)$, depending on the strategy $S$ and operators $P, Q, R$, such that
\begin{equation}\label{eq:kparitite_WeakNonSignaling}
    \begin{aligned}
        \Bigl\lvert \sum_{a_j} \sigma^{\lambda} \circ T^k_{a_k|x_k} & \cdots T^{j+1}_{a_{j+1} | x_{j+1}} \bigr( P^* T^j_{a_j | x_j}(R) Q\bigr) \\
        &-  \sum_{a_j} \sigma^{\lambda} \circ T^k_{a_k|x_k} \cdots T^{j+1}_{a_{j+1} | x_{j+1}} \bigr( P^* T^j_{a_j | x'_j}(R) Q\bigr) \Bigr\rvert \leq \negl_{S, P,R,Q}(\lambda).
    \end{aligned}
\end{equation}

Next, fix the NPA level $n \in \mathds{N}$ with word sets for all $j \in [k]$.
Consider the associated compiled moment matrix
\begin{align*}
    (\Gamma^{(n), \lambda}_\comp)_{w, v} := \sigma^{\lambda}(w^*v), \quad \Gamma^{(n), \lambda}_\comp \succeq 0,
\end{align*}
for all $w, v \in \cW^n_{[k]}$ that is ``almost'' a feasible solution of \cref{eq:kpartite_SequentialNPASDP}.
Observe that $\norm{\Gamma^{(n), \lambda}_\comp}_\op \leq \abs{\cW^n_{[k]}} $ regardless of $\lambda$.
Indeed, since $\sigma^{\lambda}$ is contractive and the generators of $C^*$-algebra $ \cA_{[k]}$ are projective, all diagonals satisfy $(\Gamma^{(n), \lambda}_\comp)_{w, w} \leq 1$, thus
\begin{align*}
    \norm{\Gamma^{(n), \lambda}_\comp}_\op = \max \text{eigenvalue of } \Gamma^{(n), \lambda}_\comp \leq \Tr{\Gamma^{(n), \lambda}_\comp} \leq \abs{\cW^n_{[k]}}.
\end{align*}
(In fact, the same argument shows that any level-$n$ normalized moment matrix on $\cW^n_{[k]}$ admits the same operator norm upper bound.)

To formalize and quantify the notion of ``almost feasibility'', we define a linear map $\cE^n_k$.
This map consolidates all linear constraints of \cref{eq:kpartite_SequentialNPASDP} by mapping a Hermitian matrix $H$ to a single vector, and $H$ satisfies these constraints if and only if it lies in the kernel of $\cE^n_k$, $\ker(\cE^n_k)$.
Specifically, let $\mathrm{Herm}(\cW^n_{[k]})$ be the space of Hermitian matrices indexed by words in $\cW^n_{[k]}$.
Let
\begin{align*}
    I_{\mathrm{Han}}^n = \{ (w, v, w', v') \mid w^*v = w'^* v' \}
\end{align*}
be the index pairs for the Hankel condition (symmetry of moment matrices) and, for each $j \in [2, k] \cap \mathds{N}$, let
\begin{align*}
    I_{j-\mathrm{ons}}^n = \{ &(a_{j+1}, \dots, a_{k}, x'_j, x_j,\dots, x_{k}, w, v, r, s) \mid \\ &\forall a_{j+1}, \dots, a_{k}, x'_j, x_j,\dots, x_{k}, \, \forall w,v\in\cW_{[j]}^{n},\, \\ &\forall r,s\in\cW_{[j-1]}^{n},\, \deg w+\deg v+\deg r+\deg s\le 2n \}
\end{align*}
be the set of tuples defining the operationally-non-signaling constraints of party $j$.
Define
\begin{equation}\label{eq:EMapForConstraint}
    \begin{aligned}
        \cE^n_k&: \mathrm{Herm}(\cW^n_{[k]}) \to \mathds{C}^{\abs{I_{\mathrm{Han}}^n}} \oplus ( \oplus_j \mathds{C}^{\abs{I_{j-\mathrm{ons}}^n}}) \\
        H &\mapsto \begin{pmatrix}
            ( H_{w,v} - H_{w',v'} )_{(w,v, w',v') \in I_{\mathrm{Han}}^n} \\
            {} \\
            \Bigl(
        \sum_{a_j}
        H\!\Bigl(
          T^k_{a_k|x_k}\cdots T^{j+1}_{a_{j+1}|x_{j+1}}
          \bigl(T^j_{a_j|x_j}(r^{*}s) - T^j_{a_j|x'_j}(r^{*}s)\bigr)
        \Bigr)_{\substack{
          T^k_{a_k|x_k}\cdots T^{j+1}_{a_{j+1}|x_{j+1}}(w),\\
          T^k_{a_k|x_k}\cdots T^{j+1}_{a_{j+1}|x_{j+1}}(v)
        }}
      \Bigr)_{{\substack{
        (a_{j+1},\dots,a_k,\, x'_j,x_j,\dots,x_k,\\
        w,v,r,s)\in I_{j\text{-}\mathrm{ons}}^n
      }}} 
        \end{pmatrix}.
    \end{aligned}
\end{equation}
Here, the top $I_{\mathrm{Han}}^n$ block of the output vector corresponds to the Hankel condition and the bottom $I_{j-\mathrm{ons}}^n$ block corresponds to the operationally-non-signaling constraints for every party $j$.
The norm of the vector $\cE^n_k(H)$ can then serve as a measure of the solution's ``almost-ness'' as formalized in \cref{lem:kpartite_ViolationConstraints}.

\begin{example}
    We give an example for the tripartite case with $k=3$ with an efficient strategy $S_3$.
    In the notation of \cref{sec:SeqNPATripartiteCase}, the weakly non-signaling conditions for Alice and Bob are, respectively,
    \begin{equation*}\label{eq:Tripartite_WeakNonSignaling}
        \begin{aligned}
            \lvert \sum_a \sigma^{\lambda}(P^* T_{\ax}(U) Q) - \sum_a \sigma^{\lambda}(P^* T_{a|x'}(U) Q) \rvert \leq \negl_{S_3, P,U,Q}(\lambda), \\
            \lvert \sum_b \sigma^{\lambda} \circ T_\ax (R^* T_{\by}(O) U) - \sum_b \sigma^{\lambda} \circ T_\ax(R^* T_{b|y'}(O) U) \rvert \leq \negl_{S_3, R,O,U}(\lambda),
        \end{aligned}
    \end{equation*}
    for all operators $P, Q \in \cA_{ABC}$, $R, U \in \cA_{BC}$, and $O \in \cA_{C}$.
    The constraint-testing map $\cE^n_3$ is
    \begin{equation*}
        \begin{aligned}
            \cE^n_3: \mathrm{Herm}(\cW^n_{ABC}) &\to \mathds{C}^{\abs{I_{\mathrm{Han}}^n}} \oplus \mathds{C}^{\abs{I_{A-\mathrm{ons}}^n}} \oplus \mathds{C}^{\abs{I_{B-\mathrm{ons}}^n}} \\
            H &\mapsto \begin{pmatrix}
                ( H_{w,v} - H_{w',v'} )_{(w,v, w',v') \in I_{\mathrm{Han}}^n} \\
                \left( H\!\Bigl(
          \sum_{a}\bigl(T_{a|x}(r^{*}s)-T_{a|x'}(r^{*}s)\bigr)
            \Bigr)_{w,v} \right)_{(x, x', w, v, r, s) \in I_{A-\mathrm{ons}}^n,} \\
            \left( H\!\Bigl(
               T_\ax\bigl(\sum_{b}T_{b|y}(t^{*}u)-T_{b|y'}(t^{*}u)\bigr)
            \Bigr)_{T_\ax(r),T_\ax(s)} \right)_{(y, y', a, x, r, s, t, u) \in I_{B-\mathrm{ons}}^n}
            \end{pmatrix},
        \end{aligned}
    \end{equation*}
    where 
    \begin{align*}
        I_{A-\mathrm{ons}}^n = \{ (x, x', w, v, r, s) \mid \forall x, x', \, \forall w,v\in\cW_{ABC}^{n},\, \forall r,s\in\cW_{BC}^{n},\, \deg w+\deg v+\deg r+\deg s\le 2n \}
    \end{align*}
    is the set of tuples defining Alice's operationally-non-signaling constraints and
    \begin{align*}
        I_{B-\mathrm{ons}}^n = \{ (y, y', a, x, r, s, t, u) \mid \forall y, y', a, x, \, \forall r,s\in\cW_{BC}^{n},\, \forall t,u\in\cW_{C}^{n},\, \deg r+\deg s+\deg t+\deg u\le 2n \}
    \end{align*}
    is the corresponding set for Bob's.
\end{example}

\subsection{A closeness result via geometric arguments}\label{sec:GeometricProof}
In this subsection, we show the main technical result (\cref{thm:kpartite_BoundToFeasibleSolution}, graphically \cref{fig:GeometryOfProof}), which states that for any compiled strategy $S$ and any $n$, there exists a feasible solution to our level-$n$ sequential NPA hierarchy \cref{eq:kpartite_SequentialNPASDP} that is negligibly close to $S$.

To this end, let us first formalize the notation of ``almost feasibility'' of the compiled moment matrix $\Gamma^{(n), \lambda}_\comp$ as follows.
\begin{lemma}\label{lem:kpartite_ViolationConstraints}
    For the compiled moment matrix $\Gamma^{(n), \lambda}_\comp$ extracted from an efficient (quantum polynomial-time) strategy $S = (S_{\lambda})$ for a $k$-partite compiled nonlocal game, there exists a negligible function $\negl^{\mathrm{lem}}_{S, n}(\lambda)$, dependent on $S, n$, such that
    \begin{align*}
        \norm{\cE^n_k(\Gamma^{(n), \lambda}_\comp)}_2 \leq \negl^{\mathrm{lem}}_{S, n}(\lambda),
    \end{align*}
    where $\norm{\cdot}_2$ denotes the usual Euclidean norm.
\end{lemma}
\begin{proof}
    Clearly $(\Gamma^{(n), \lambda}_\comp)_{w,v} - (\Gamma^{(n), \lambda}_\comp)_{w',v'} = 0$ if $w^*v = w'^* v'$ by its definition.
    On the other hand, by \cref{eq:kparitite_WeakNonSignaling},
    \begin{align*}
        \Bigl\lvert
            &\sum_{a_j}
            \Gamma^{(n), \lambda}_\comp \!\Bigl(
              T^k_{a_k|x_k}\cdots T^{j+1}_{a_{j+1}|x_{j+1}}
              \bigl(T^j_{a_j|x_j}(r^{*}s) - T^j_{a_j|x'_j}(r^{*}s)\bigr)
            \Bigr)_{\substack{
              T^k_{a_k|x_k}\cdots T^{j+1}_{a_{j+1}|x_{j+1}}(w),\\
              T^k_{a_k|x_k}\cdots T^{j+1}_{a_{j+1}|x_{j+1}}(v)
            }}
          \Bigr\rvert \\
          &= \abs{ \sum_{a_j} \sigma^{\lambda} \circ T^k_{a_k|x_k} \cdots T^{j+1}_{a_{j+1} | x_{j+1}} \bigr( w^* T^j_{a_j | x_j}(r^*s) v\bigr) -  \sum_{a_j} \sigma^{\lambda} \circ T^k_{a_k|x_k} \cdots T^{j+1}_{a_{j+1} | x_{j+1}} \bigr( w^* T^j_{a_j | x'_j}(r^* s) v\bigr) } \\
          &\leq \negl_{S, w,r,s,v}(\lambda) \leq \max_{w,r,s,v}\negl_{S, w,r,s,v}(\lambda) := \negl_{S,n}'(\lambda),
    \end{align*}
    where the maximum makes sense because the degree $n$ word sets $\cW^n_{[j-1]}, \cW^n_{[j]}$ are finite.
    Then, as the $I_{\mathrm{Han}}^n$-block is already $0$ for $\cE^n_k(\Gamma^{(n), \lambda}_\comp)$, one has
    \begin{align*}
        \norm{\cE^n_k(\Gamma^{(n), \lambda}_\comp)}_2 \leq \sqrt{\sum_{j} \abs{I_{j-\mathrm{ons}}^n}} \cdot \negl_{S, n}'(\lambda) := \negl^{\mathrm{lem}}_{S, n}(\lambda).
    \end{align*}
\end{proof}

\begin{theorem}\label{thm:kpartite_BoundToFeasibleSolution}
    Let $\cG$ be a $k$-partite nonlocal game with $\cG_{\comp}$ as its compiled version.
    Let $S = (S_{\lambda})_{\lambda}$ be an arbitrary efficient (quantum polynomial-time) strategy employed by the prover, and consider the corresponding algebraic compiled strategy state $\sigma^{\lambda}: \cA_{[k]} \to \mathds{C}$ due to~\cite{baroni2025asymptotic}.
    For every $n \in \mathds{N}$, let $(\Gamma^{(n), \lambda}_\comp)_{w,v} := \sigma^{\lambda}(w^*v)$ be the associated level-$n$ moment matrix.
    
    Then, there exists a $(S, n)$-dependent negligible function $\negl_{S,n}(\lambda)$ (goes to zero faster than the reciprocal of any polynomial in $\lambda$) and a feasible solution $\Gamma^{(n), \lambda}$ of \cref{eq:kpartite_SequentialNPASDP} such that
    \begin{align*}
        \norm{ \Gamma^{(n), \lambda}_\comp - \Gamma^{(n), \lambda}}_\op \leq \negl_{S, n}(\lambda),
    \end{align*}
    where $\norm{\cdot}_\op$ denotes the operator spectral norm.
\end{theorem}
\begin{proof}
    We aim to construct a feasible solution $\Gamma^{(n), \lambda}$ that is negligibly close to the compiled moment matrix $\Gamma^{(n), \lambda}_\comp$, see~\cref{fig:GeometryOfProof} for a pictorial illustration.
    To this end, consider the projection onto $\ker(\cE^n_k)$
    \begin{align*}
        \Pi := \id - (\cE^n_k)^\dagger \cE^n_k,
    \end{align*}
    where $(\cE^n_k)^\dagger$ denotes the Moore-Penrose pseudo-inverse of $\cE^n_k$ (see e.g.,~\cite{golub2013matrix}).
    Thus, the projection of the compiled moment matrix
    \begin{align*}
        \Gamma_1^{(n), \lambda} := \Pi(\Gamma^{(n), \lambda}_\comp)
    \end{align*}
    satisfies all linear constraints of \cref{eq:kpartite_SequentialNPASDP} and
    \begin{equation}\label{eq:kpartite_ThmProof_ProjectedMatrixBound}
    \begin{aligned}
        \norm{ \Gamma_1^{(n), \lambda} - \Gamma^{(n), \lambda}_\comp }_{\op} & = \norm{ (\cE^n_k)^\dagger \cE^n_k( \Gamma^{(n), \lambda}_\comp ) }_{\op} \leq \norm{(\cE^n_k)^\dagger}_\op \norm{\cE^n_k(\Gamma^{(n), \lambda}_\comp)}_2 \\ & \leq \norm{(\cE^n_k)^\dagger}_\op \negl^{\mathrm{lem}}_{S, n}(\lambda) := \negl_{S, n}'(\lambda),
    \end{aligned}
    \end{equation}
    where we use \cref{lem:kpartite_ViolationConstraints} and absorb the bounded constant $\norm{(\cE^n_k)^\dagger}_\op$ into the negligible function.
    
    While $\Gamma_1^{(n), \lambda}$ satisfies the Hankel condition and operationally non-signaling constraints of \cref{eq:kpartite_SequentialNPASDP}, it is no longer normalized nor necessarily positive semidefinite.
    Nonetheless, by the fact that $\Gamma^{(n), \lambda}_\comp \succeq 0$ and \cref{eq:kpartite_ThmProof_ProjectedMatrixBound}, Weyl's inequality~\cite{weyl1912asymptotische} implies that the minimum eigenvalue $\mu_0(\Gamma_1^{(n), \lambda})$ of $\Gamma_1^{(n), \lambda}$ satisfies
    \begin{align}\label{eq:kparitite_ThmProof_MinEigen}
        \mu_0(\Gamma_1^{(n), \lambda}) \geq - \negl_{S, n}'(\lambda),
    \end{align}
    i.e., $\Gamma_1^{(n), \lambda}$ is almost positive semidefinite.
    Furthermore, normalization is almost fulfilled in the sense that
    \begin{equation}\label{eq:kparitite_ThmProof_NormalizationBound}
    \begin{aligned}
        \abs{ (\Gamma_1^{(n), \lambda})_{\id,\id} - 1 } & = \abs{ (\Gamma_1^{(n), \lambda})_{\id,\id} - (\Gamma^{(n), \lambda}_\comp)_{\id,\id} } \leq \norm{\Gamma_1^{(n), \lambda} - \Gamma^{(n), \lambda}_\comp}_{\max}\\ & \leq \norm{ \Gamma_1^{(n), \lambda} - \Gamma^{(n), \lambda}_\comp }_{\op} \leq \negl_{S, n}'(\lambda),
    \end{aligned}
    \end{equation}
    where $\norm{\cdot}_{\max}$ denotes the max norm and $\norm{\cdot}_{\max} \leq \norm{\cdot}_{\op}$ is immediate.
    The last observation is that $\norm{\Gamma^{(n), \lambda}_1}_\op \leq \norm{\Gamma^{(n), \lambda}_\comp}_\op \leq \abs{\cW^n_{[k]}}$, independent of $\lambda$, due to the contractivity of $\Pi$.

    Next, consider the strictly feasible solution $\Gamma^{(n)}_\strict$ of \cref{prop:kpartite_StrictFeasible} and denote its minimal eigenvalue by $\mu_n > 0$.
    With a convex combination of $\Gamma^{(n)}_\strict$ and $\Gamma_1^{(n), \lambda}$, define 
    \begin{align*}
        \Gamma^{(n), \lambda}_2 := \frac{\mu_n}{\mu_n + \negl_{S, n}'(\lambda)} \Gamma_1^{(n), \lambda} + \frac{\negl_{S, n}'(\lambda)}{\mu_n + \negl_{S, n}'(\lambda)}\Gamma^{(n)}_\strict.
    \end{align*}
    It follows from \cref{eq:kparitite_ThmProof_MinEigen} and Weyl's inequality~\cite{weyl1912asymptotische} that the minimal eigenvalue $\mu_0(\Gamma^{(n), \lambda}_2)$ of $\Gamma^{(n), \lambda}_2$ satisfies
    \begin{align*}
        \mu_0(\Gamma^{(n), \lambda}_2) \geq \frac{\mu_n}{\mu_n + \negl_{S, n}'(\lambda)} \underbrace{\mu_0 ( \Gamma_1^{(n), \lambda} )}_{ \geq - \negl_{S, n}'(\lambda)} + \frac{\negl_{S, n}'(\lambda)}{\mu_n + \negl_{S, n}'(\lambda)} \underbrace{\mu_0 (\Gamma^{(n)}_\strict )}_{=\mu_n} \geq 0,
    \end{align*}
    consequently $\Gamma^{(n), \lambda}_2 \succeq 0$.
    Thus, the matrix $\Gamma^{(n), \lambda}_2$ satisfies all constraints of \cref{eq:kpartite_SequentialNPASDP} except for the normalization, because of the convexity of the constraint set.
    Moreover,
    \begin{align}\label{eq:kparitite_ThmProof_Gamma1Gamma2Bound}
        \norm{\Gamma^{(n), \lambda}_1 - \Gamma^{(n), \lambda}_2}_\op = \frac{\negl_{S, n}'(\lambda)}{\mu_n + \negl_{S, n}'(\lambda)} \norm{\Gamma^{(n), \lambda}_1 - \Gamma^{(n)}_\strict}_\op \leq \frac{\negl_{S, n}'(\lambda)}{\mu_n} \norm{\Gamma^{(n), \lambda}_1 - \Gamma^{(n)}_\strict}_\op := \negl_{S, n}''(\lambda),
    \end{align}
    where the constant $\norm{\Gamma^{(n), \lambda}_1 - \Gamma^{(n)}_\strict}_\op/\mu_n$ is absorbed into the negligible function, since both $\norm{\Gamma^{(n), \lambda}_1}_\op$ and $\norm{\Gamma^{(n)}_\strict}_\op$ are upper bounded by $\abs{\cW^n_{[k]}}$ that is independent of $\lambda$.
    Note that by the definition of $\Gamma_2^{(n), \lambda}$, $\mu_n > 0$, $\negl_{S, n}'(\lambda) \geq 0$, and \cref{eq:kparitite_ThmProof_NormalizationBound},
    \begin{align}\label{eq:kparitite_ThmProof_NormalizationGamma2}
        \abs{(\Gamma_2^{(n), \lambda})_{\id,\id} - 1} = \frac{\mu_n}{\mu_n + \negl_{S, n}'(\lambda)} \abs{(\Gamma_1^{(n), \lambda})_{\id,\id} - 1} \leq \frac{\mu_n \negl_{S, n}'(\lambda)}{\mu_n + \negl_{S, n}'(\lambda)} \leq \negl_{S, n}'(\lambda).
    \end{align}
    
    Finally, we normalize $\Gamma^{(n), \lambda}_2$ and obtain 
    \begin{align*}
        \Gamma^{(n), \lambda} := {\Gamma^{(n), \lambda}_2}/{(\Gamma^{(n), \lambda}_2)_{\id, \id}},
    \end{align*}
    which is a proper feasible solution of \cref{eq:kpartite_SequentialNPASDP} since all constraints (that are not normalization) are homogeneous and thus preserved by this renormalization.
    It follows from \cref{eq:kpartite_ThmProof_ProjectedMatrixBound,eq:kparitite_ThmProof_Gamma1Gamma2Bound,eq:kparitite_ThmProof_NormalizationGamma2} that
    \begin{align*}
        \norm{\Gamma^{(n), \lambda}_\comp - \Gamma^{(n), \lambda}}_\op &\leq \norm{\Gamma^{(n), \lambda}_\comp - \Gamma^{(n), \lambda}_1}_\op + \norm{\Gamma^{(n), \lambda}_1 - \Gamma^{(n), \lambda}_2}_\op + \norm{\Gamma^{(n), \lambda}_2 - \Gamma^{(n), \lambda}}_\op \\
        &\leq \negl_{S, n}'(\lambda) + \negl_{S, n}''(\lambda) + \abs{1 - \frac{1}{(\Gamma^{(n), \lambda}_2)_{\id, \id}}} \norm{\Gamma^{(n), \lambda}_2}_\op \\
        &\leq \negl_{S, n}'(\lambda) + \negl_{S, n}''(\lambda) + \abs{\frac{\norm{\Gamma^{(n), \lambda}_2}_\op}{(\Gamma^{(n), \lambda}_2)_{\id, \id}}} \negl_{S, n}'(\lambda) := \negl_{S, n}(\lambda),
    \end{align*}
    where $\negl_{S, n}(\lambda)$ is a negligible function due to the fact that $\norm{\Gamma^{(n), \lambda}_2}_\op/\abs{(\Gamma^{(n), \lambda}_2)_{\id, \id}}$ is a bounded constant.
\end{proof}

\subsection{Putting the puzzle pieces together---quantitative quantum soundness}\label{sec:PiecesTogetherCorollary}
Putting together the existence and closeness statement from \cref{thm:kpartite_BoundToFeasibleSolution} and the convergence of the sequential NPA hierarchy \cref{thm:kpartite_Completeness}, we obtain, as a direct consequence, the following quantitative quantum soundness statement for multipartite compiled nonlocal games.
(We refer to \cref{fig:big-scheme} for a graphical representation of the buildup of the following main result.)

\begin{corollary}\label{cor:kpartite_QuantumSoundness}
    Let $\cG$ be a $k$-partite nonlocal game with $\cG_{\comp}$ as its compiled version.
    Let $S = (S_{\lambda})_{\lambda}$ be an arbitrary efficient (quantum polynomial-time) strategy employed by the prover.
    Then for every $n \geq 1$, there exists a negligible function $\negl_{S, n}(\lambda)$ (dependent on the QHE scheme, the strategy S, and the NPA level $n$) such that
    \begin{align}\label{eq:CorollaryGeneral}
        \gamevalueCompile{\lambda}{\cG_{\comp}, S} \leq \gamevaluekSeqNPA{\cG}{n} + \negl_{S, n}(\lambda),
    \end{align}
    where $\gamevalueCompile{\lambda}{\cG_{\comp}, S}$ is the prover's Bell score using $S$ and $\gamevaluekSeqNPA{\cG}{n}$ is the optimal value of the hierarchy \cref{eq:kpartite_SequentialNPASDP} at level $n$.
    
    Furthermore, if the $k$-partite sequential NPA hierarchy for $\cG$ converges at some finite level $n'$, i.e., $\gamevaluekSeqNPA{\cG}{n'} = \omega_{\mathrm{qc}}(\cG)$, then there exists a negligible function $\negl_{S}(\lambda)$ (dependent on the QHE scheme, the strategy S) such that
    \begin{align}\label{eq:CorollaryFiniteConverge}
        \gamevalueCompile{\lambda}{\cG_{\comp}, S} \leq \omega_{\mathrm{qc}}(\cG) + \negl_{S}(\lambda),
    \end{align}
    where $\omega_{\mathrm{qc}}(\cG)$ is the optimal commuting quantum score.
    In particular, if the hierarchy admits a flat optimal solution at some finite level, then $\gamevalueqcopt{\cG}=\omega_{\mathrm q}(\cG)$ where $\omega_{\mathrm{q}}(\cG)$ is the optimal tensor product quantum score.
    Hence
    \begin{align}\label{eq:CorollaryFlat}
        \gamevalueCompile{\lambda}{\cG_{\comp}, S} \leq \omega_{\mathrm{q}}(\cG) + \negl_{S}(\lambda).
    \end{align}
\end{corollary}
\begin{proof}
    Let $\sigma^{\lambda}: \cA_{[k]} \to \mathds{C}$ be the corresponding algebraic compiled strategy state due to~\cite{baroni2025asymptotic}.
    For each $n \geq 1$, consider the associated compiled moment matrix
    \begin{align*}
        (\Gamma^{(n), \lambda}_\comp)_{w, v} := \sigma^{\lambda}(w^*v) \succeq 0
    \end{align*}
    for all $w, v \in \cW^n_{[k]}$.
    
    By \cref{thm:kpartite_BoundToFeasibleSolution}, there exists an $n$-dependent negligible function $\negl_{S, n}'(\lambda)$ and a feasible solution $\Gamma^{(n), \lambda}$ for every $\lambda$ of the $k$-partite sequential NPA hierarchy of \cref{eq:kpartite_SequentialNPASDP} such that
    \begin{align*}
        \norm{ \Gamma^{(n), \lambda}_\comp - \Gamma^{(n), \lambda}}_\op \leq \negl_{S, n}'(\lambda).
    \end{align*}
    Then,
    \begin{align*}
        \gamevalueCompile{\lambda}{\cG_{\comp}, S}
        &= \sigma^{\lambda}(\beta)
         = \sum_{a_{[k]}, x_{[k]}} \beta_{a_{[k]} x_{[k]}} (\Gamma^{(n), \lambda}_\comp)_{\id, f_\akxk} \\
        &= \underbrace{%
             \sum_{a_{[k]}, x_{[k]}} \beta_{a_{[k]} x_{[k]}} \Gamma^{(n), \lambda}_{\id, f_\akxk}%
           }_{\text{score of a possibly non-optimal feasible solution}} 
        +\quad{} \sum_{a_{[k]}, x_{[k]}} \beta_{a_{[k]} x_{[k]}}
              \left(
                \underbrace{%
                  (\Gamma^{(n), \lambda}_\comp)_{\id, f_\akxk}
                  - \Gamma^{(n), \lambda}_{\id, f_\akxk}%
                }_{\text{bounded by the max norm}}%
              \right) \\
        &\leq \gamevaluekSeqNPA{\cG}{n}
           + \underbrace{\abs{\sum_{a_{[k]}, x_{[k]}} \beta_{a_{[k]} x_{[k]}}}}_{\text{game-related constant}}
             \left\lVert \Gamma^{(n), \lambda}_\comp - \Gamma^{(n), \lambda} \right\rVert_{\max} \\
        &\leq \gamevaluekSeqNPA{\cG}{n}
           + \mathrm{const} \cdot
             \left\lVert \Gamma^{(n), \lambda}_\comp - \Gamma^{(n), \lambda} \right\rVert_{\op} \\
        &\leq \gamevaluekSeqNPA{\cG}{n}
           + \mathrm{const} \cdot \negl_{S, n}'(\lambda)
         = \gamevaluekSeqNPA{\cG}{n} + \negl_{S,n}(\lambda),
    \end{align*}
    where $\negl_{S, n}(\lambda) := \mathrm{const} \cdot \negl_{S, n}'(\lambda)$ is again negligible.
    
    Finally, if the hierarchy converges at level $n'$, we apply the first part with $n = n'$ and set $\negl_{S}(\lambda) := \negl_{S, n'}(\lambda)$.
    If the hierarchy admits a flat optimal solution at level $n'$, then \cref{prop:kpartite_FlatnessCondition} further gives $\gamevaluekSeqNPA{\cG}{n'} = \omega_{\mathrm{qc}}(\cG) = \omega_{\mathrm q}(\cG)$.
\end{proof}

\begin{figure}
    \centering
    \includegraphics[width=0.95\linewidth]{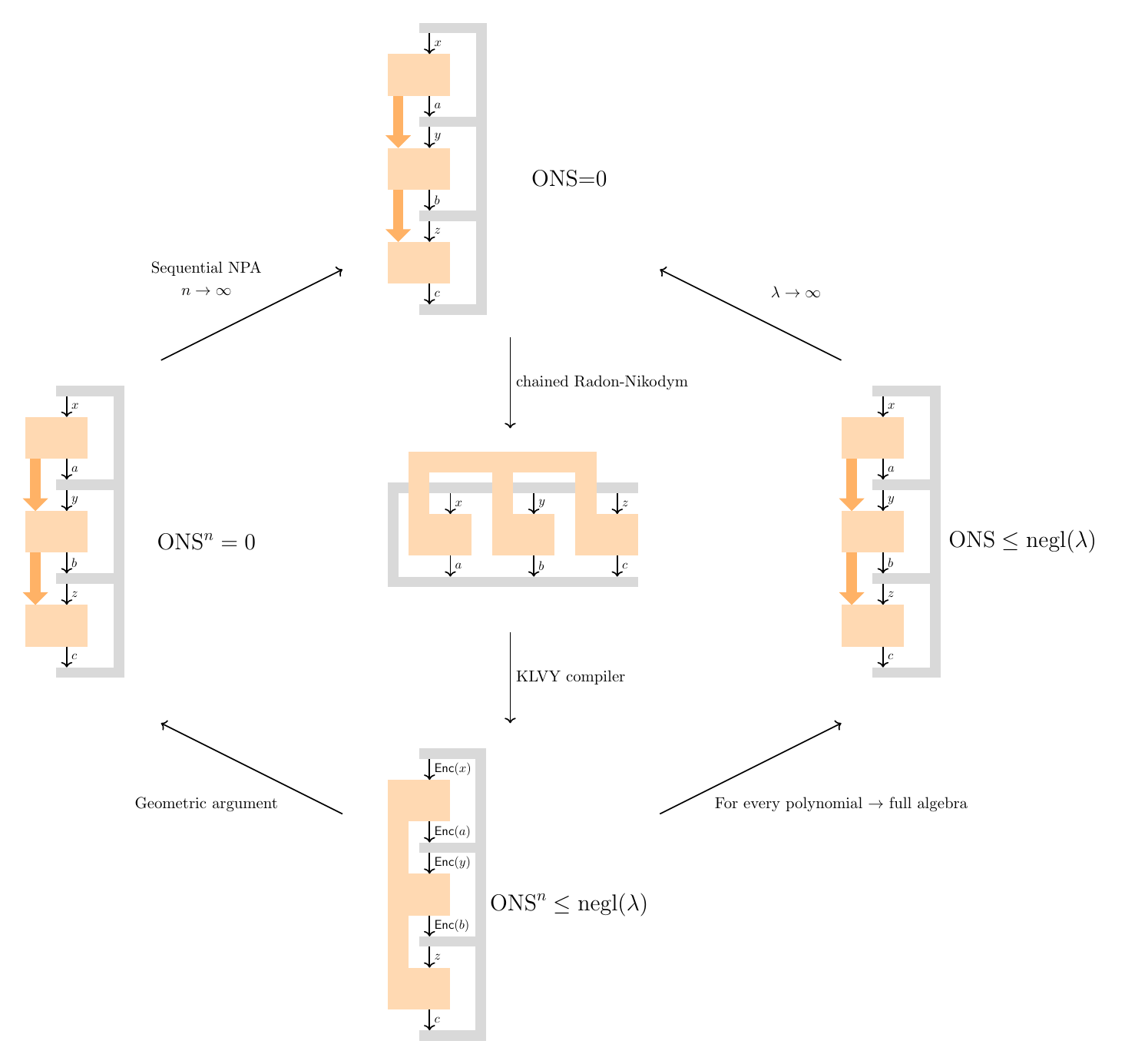}
    \caption{Graphical overview of the proof structure of this paper (left loop) compared to the proof structure of \cite{baroni2025asymptotic}(right loop).
    Starting from the center, we consider a tripartite non-local game. This can be compiled into a single-prover interactive game through the KLVY compiler (downward arrow). Using the block-encoding arguments introduced in \cite{natarajan2023bounding,kulpe2024bound,baroni2025asymptotic} it is possible to show that an efficient compiled single-prover satisfy weakly (with respect to $\lambda$) the operational-non-signalling constraint when tested against polynomials of degrees at most $n$; this is what we mean by $\text{ONS}^n\leq \negl(\lambda)$.
    Proving the quantum soundness is equivalent to show that this arrow can be inverted, meaning that all efficient compiled strategies can be mapped back to a non-local strategy.
    In the asymptotic approaches of \cite{kulpe2024bound} and \cite{baroni2025asymptotic}, the authors consider that this statement is true for all polynomials of any degrees $n$. Since this is a dense subset of the algebra, they show a $\lambda$-weak operational-non-signalling constraint on the full algebra. Taking the limit of $\lambda \to \infty$ then yields to perfect operational-non-signalling constraints. Finally, the Radon-Nikodym theorem (and its chain-rule version proved in \cite{baroni2025asymptotic}) implies that this setting is equivalent to commuting operator strategies. 
    In this work, we present a compatible but different approach. First, we use a geometric argument to get rid of the $\lambda$ dependency, by projecting the correlations from the compiled game to correlations that satisfy perfectly operational-non-signalling when tested against polynomials of degree $n$, by committing an error negligible in $\lambda$ that we can keep track of. Next, we develop a sequential NPA hierarchy, and we show that in the limit of $n \to \infty$ it converges to a sequential scenario with perfect operational-non-signalling on the full algebra, reconciling with the previous asymptotic proofs.
    This has the advantage that, unlike previous asymptotic proofs, it provides concrete statements for finite levels of the security parameter $\lambda$, which is of great technical importance for implementations.}
    \label{fig:big-scheme}
\end{figure}

Note that we recover~\cite[Theorem~A]{klep2025quantitative} the quantitative quantum upper bound statement for all bipartite compiled nonlocal games by setting $k=2$, and the asymptotic results of~\cite{baroni2025asymptotic} by letting $\lambda, n \to \infty$ thanks to \cref{thm:kpartite_Completeness}.
As in the bipartite case~\cite{klep2025quantitative}, flat optimality also extracts a finite-dimensional optimal strategy; combined with the quantum completeness of the KLVY compiler~\cite{kalai2023quantum}, this gives a matching lower bound up to negligible error whenever the extracted strategy satisfies the usual implementation requirements.

\begin{remark}\label{rem:NegligibleFunctionIndependent}
A useful feature of \cref{cor:kpartite_QuantumSoundness} is that the error terms $\negl_{S,n}(\lambda)$ are determined solely by the QHE scheme underlying the compiled game $\cG_{\comp}$ and by the given efficient compiled strategy $S$.
In particular, their definition and negligibility do \emph{not} rely on our
ability to solve (or certify convergence of) the sequential hierarchy.

Concretely, the function $\negl^{\mathrm{lem}}_{S,n}(\lambda)$ in \cref{lem:kpartite_ViolationConstraints} arises from bounding the maximal constraint violation over the finitely many monomials of degree at most $n$, together with an $n$-dependent factor reflecting the size of the level-$n$
moment matrix for $\cG$.
For each fixed monomial, the corresponding negligible bound depends only on $S$ and the security of the underlying QHE scheme (see, e.g., \cite[Theorem~11]{baroni2025asymptotic}).
Tracing the proof of \cref{thm:kpartite_BoundToFeasibleSolution}, the final $\negl_{S,n}(\lambda)$ differs from $\negl^{\mathrm{lem}}_{S,n}(\lambda)$ only by additional explicit $n$-dependent constants coming from the projection/PSD-repair steps.
\end{remark}

We remark further that one cannot simply remove the $n$-dependence of \cref{eq:CorollaryGeneral} by taking $n \to \infty$ for quantitative quantum soundness.
Indeed, while one can obtain
\begin{align*}
        \gamevalueCompile{\lambda}{\cG_{\comp}, S} \leq \omega_{\mathrm{qc}}(\cG) + \lim_{n \to \infty}\negl_{S, n}(\lambda),
    \end{align*}
but the function $\lim_{n \to \infty} \negl_{S, n}(\lambda)$ is not necessarily negligible in general.
As discussed in \cref{rem:NegligibleFunctionIndependent}, the behavior of this limit is highly dependent on the efficient compiled strategy $S$ and the size of the game $\cG$.

\section*{Acknowledgments}
We thank Michael Walter for the helpful discussions.
MB acknowledges funding from QuantEdu France, a state aid managed by the French National Research Agency for France 2030 with the reference ANR-22-CMAS-0001.
IK was supported by the Slovenian Research and Innovation Agency program P1-0222 and grants J1-50002, N1-0217, J1-3004, J1-50001, J1-60011, J1-60025.
Partially supported by the Fondation de l'École polytechnique as part of the Gaspard Monge Visiting Professor Program.
IK thanks École Polytechnique and Inria for hospitality during the preparation of this manuscript. 
DL acknowledges support from the Quantum Advantage Pathfinder (QAP) research program within the UK’s National Quantum Computing Center (NQCC).
MOR, LT, and XX acknowledge funding by the ANR for the JCJC grant LINKS (ANR-23-CE47-0003) and T-ERC QNET (ANR-24-ERCS-0008), by INRIA and CIEDS in the Action Exploratoire project DEPARTURE. MOR, IK, LT, and XX acknowledge support by the European Union's Horizon 2020 Research and Innovation Programme under QuantERA Grant Agreement no. 731473 and 101017733.

\printbibliography

\appendix
\section{Appendix}\label{sec:Appendix}
We give a brief comparison of the two related previous works~\cite{klep2025quantitative,baroni2025asymptotic}.

\subsection{%
\texorpdfstring{The bipartite sequential NPA hierarchy of~\cite{klep2025quantitative}}%
{The bipartite sequential NPA hierarchy of [Kle+25]}%
}%
\label{sec:PrelimSeqNPA}

As the first attempt to bridge the gap between the standard NPA hierarchy and the sequential protocols central to this work (\cref{fig:NonlocalCompiledBellGame}), the authors of~\cite{klep2025quantitative} introduced a sequential variant of the NPA hierarchy for \emph{bipartite} games.
Different from our Heisenberg picture inspired hierarchy in \cref{sec:SeqNPAMain}, they model the scenario from the \emph{Schr\"odinger picture}: after Alice receives question $x$ and produces answer $a$, the shared state $\sigma$ collapses into a new, subnormalized state $\sigma_{a|x}$ for Bob, who can now perform a measurement $y$ and obtain an output $b$.
Instead of a single moment matrix representing the global pre-measurement state, the hierarchy is defined via a collection of moment matrices labeled by the actions of Alice $\Theta^{(n)}(a|x)$, which are indexed by words built from Bob's operators $f_{b|y}$.

The constraint of operationally-non-signaling translates to this family of moment matrices as follows.
\begin{definition}[Bipartite sequential NPA hierarchy~\cite{klep2025quantitative}]\label{def:Prelim_BipartiteSequentialNPASDP}
    For a bipartite game $\cG$, the level-$n$ sequential NPA relaxation is the solution to the following SDP, defined in terms of a collection of subnormalized moment matrices $\Theta^{(n)}(a|x)$ indexed by words of length $\leq n$ in letters $\{f_{b|y}\}$:
    \begin{equation}\label{eq:Prelim_BipartiteSequentialNPASDP}
        \begin{aligned}
            \gamevalueSeqNPA{\cG}{n} \quad &= \max_{\Theta^{(n)}(a|x)} \sum_{a,b,x,y} \beta_{abxy}\Theta^{(n)}(a|x)_{1, f_{b|y}} \\
            \text{s.t.} \quad 
            &\Theta^{(n)}(a|x) \succeq 0, \quad \forall a,x, \\[3pt]
            & \sum_{a} \Theta^{(n)}(a|x) = \sum_{a} \Theta^{(n)}(a|x') := \Theta^{(n)} \quad \forall x, x' \quad \text{(operationally-non-signaling condition)} \\
            & 1 = \Theta^{(n)}_{\id, \id} \qquad \text{(normalization)}.
        \end{aligned}
    \end{equation}
\end{definition}

This hierarchy effectively characterizes the set of bipartite sequential quantum correlations and the commuting observable quantum correlations as $n \to \infty$, and possesses several key properties:
\begin{enumerate}[label=(\alph*)]
    \item \emph{Soundness and monotone convergence.}
    Like the standard NPA hierarchy from \cref{def:Prelim_StandardNPA}, it is sound,
    \begin{align*}
        \omega_{\mathrm{qc}}(\cG) \leq \gamevalueSeqNPA{\cG}{n},
    \end{align*}
    and converges to the commuting-operator value monotonically,
    \begin{align*}
        \gamevalueSeqNPA{\cG}{1} \geq \gamevalueSeqNPA{\cG}{2} \geq \cdots \searrow \omega_{\mathrm{qc}}(\cG).
    \end{align*}
    \item \emph{Relation to standard NPA.} At any finite level $n$, it is equivalent to a relaxed version of the standard NPA hierarchy where Alice's operators $f_{a|x}$ do not satisfy the POVM completeness condition: $\sum_a f_{a|x} \neq \id$.
    It only need to appear to be complete when testing with length $\leq n$ words of Bob's operators $f_{b|y}$.
    \item \emph{Duality.} Its conic dual problem is equivalent to a sparse sum-of-squares (SOS) hierarchy~\cite{klep2022sparse,magron2023sparse}.
    Similar to its dual counterpart, the bipartite sequential NPA hierarchy can be less computationally demanding, as the moment matrices $\Theta^{(n)}(a|x)$ indexed by only $f_{b|y}$ are much smaller than the standard bipartite NPA moment matrices.
\end{enumerate}

However, as discussed in \cref{sec:Introduction}, this Schr\"odinger-picture construction based on post-measurement states is \emph{inherently bipartite}.
It does not naturally extend to the multipartite setting, as tracking the sequence of post-measurement states loses too much of the necessary algebraic structure.

\subsection{%
\texorpdfstring{The multipartite algebraic framework of~\cite{baroni2025asymptotic}}%
{The multipartite algebraic framework of [Bar+25]}%
}%
\label{sec:PrelimMulti}

In order to extend the operator-algebraic techniques of \cite{kulpe2024bound} from the bipartite to the multipartite setting, the authors of \cite{baroni2025asymptotic} propose a composable algebraic structure that is very well-suited to characterize sequential players.
For the notation convenience, let us focus on the tripartite scenarios.
Their key idea is the new concept of universal C$^*$-algebras of sequential projective measurements, which generalizes the classical universal algebra of static measurements, i.e., what before we were referring to as the "symbols" for Bob $f_{b|y}$.
More concretely, let us consider the non-trivial case of three players in a sequence.
The action of Alice is modeled as an algebraic state $\phi_{a|x}$, Bob performs a transformation and in the end Charlie measures; or equivalently in the Heisenberg picture, Charlie performs a measurement and Bob implements a transformation that is pulling back Bob's measurement to the space in which Alice's state is defined.
In this picture, Charlie and Bob jointly perform a measurement, whose only constraint is that Bob acts before Charlie, hence the latter cannot signal to the first, while the contrary is allowed.

In more detail, \cite{baroni2025asymptotic} define the universal C$^*$-algebra $\mathcal{A}_{B \to C}$ of Bob's and Charlie's sequential measurements using the following relations on its generators $\{f_\bcyz\}$:
\begin{align*}
    f_\bcyz^* = f_\bcyz,
    \quad f_\bcyz f_{b'c'|yz} = \delta_{b,b'} \delta_{c,c'} f_\bcyz,
    \quad \sum_{b,c} f_\bcyz = \id,
    \quad \sum_c f_\bcyz = \sum_c f_{bc|yz'}, \; \forall z, z'.
\end{align*}
Note that $\mathcal{A}_{B \to C}$ is precisely $\cA_{BC}$ that we introduce in \cref{sec:SeqNPATripartiteCase}.
Letting $\mathcal{A}_C$ be the universal C$^*$-algebra of Charlie's measurements, there exists then a family of *-homomorphisms $T_{\by} : \mathcal{A}_C \to \mathcal{A}_{B \to C}$, taking generators to generators naturally in the following way
\begin{align*}
    T_{\by}(f_\cz) = f_\bcyz.
\end{align*}
Note, that by construction, summing over $b$ yields unital *-homomorphisms $T_y = \sum_b T_\by$.
In this way, the maps $T_\by$ can be thought of as embeddings of $\mathcal{A}_C$ into subalgebras of $\mathcal{A}_{B \to C}$ labeled by $b$ and $y$.

The power of this approach lies in two important features:
\begin{enumerate}
    \item \textit{Composability}, it is very immediate to understand how to increase the number of sequential players: it is sufficient to add the labels and the additional "causal" constraint;
    \item It provides a very \textit{compact framework}, because all of the players strategies are defined on a single big algebra, with a rich internal structure that retrieves the action of the single players, which is captured by *-homomorphisms.
\end{enumerate}

\end{document}